\documentclass[aip,jcp,preprint,footinbib]{revtex4-1}

\usepackage{amsmath}
\usepackage{amssymb}
\usepackage{amsthm}
\usepackage{amsfonts}
\usepackage{graphicx}
\usepackage{dsfont}
\usepackage{hyperref}

\graphicspath{ {../figures/} }

\usepackage[table]{xcolor}

\DeclareMathOperator*{\maximize}{maximize}
\DeclareMathOperator*{\argmax}{argmax}

\DeclareMathOperator*{\diag}{diag}
\DeclareMathOperator*{\R}{\mathbb{R}}
\DeclareMathOperator*{\E}{\mathbb{E}}

\newenvironment{my_algorithm}{
	\fontfamily{\ttdefault} \selectfont
	\begin{enumerate}
		\setlength{\itemsep}{1pt}
		\setlength{\parskip}{0pt}
		\setlength{\parsep}{0pt}
}{
	\end{enumerate}
}

\begin{document}

\definecolor{review_red}{rgb}{0.9, 0.0, 0.0}
\definecolor{delred}{rgb}{1.0, 0.6, 0.6}
\definecolor{addgreen}{rgb}{0.0, 0.6, 0.0}
\newcommand{\tradd}[1]{#1}
\newcommand{\trdel}[1]{}
\newcommand{\trred}[1]{#1}

\title{Adaptive Hybrid Simulations for Multiscale Stochastic Reaction Networks}

\author{Benjamin Hepp}
\author{Ankit Gupta}
\author{Mustafa Khammash}
\affiliation{Department of Biosystems Science and Engineering (D-BSSE), ETH Zurich Switzerland}

\date{\today}

\begin{abstract}
The probability distribution describing the state of a Stochastic Reaction Network evolves according to the Chemical Master Equation (CME). It is common to estimated its solution using Monte Carlo methods such as the Stochastic Simulation Algorithm (SSA). In many cases these simulations can take an impractical amount of computational time. Therefore many methods have been developed that approximate the Stochastic Process underlying the Chemical Master Equation. Prominent strategies are Hybrid Models that regard the firing of some reaction channels as being continuous and applying the quasi-stationary assumption to approximate the dynamics of fast subnetworks. However as the dynamics of a Stochastic Reaction Network changes with time these approximations might have to be adapted during the simulation. We develop a method that approximates the solution of a CME by automatically partitioning the reaction dynamics into discrete/continuous components and applying the quasi-stationary assumption on identifiable fast subnetworks. Our method does not require user intervention and it adapts to exploit the changing timescale separation between reactions and/or changing magnitudes of copy numbers of constituent species. We demonstrate the efficiency of the proposed method by considering examples from Systems Biology and showing that very good approximations to the exact probability distributions can be achieved in significantly less computational time.
\end{abstract}

\providecommand{\keywords}[1]{\textbf{Keywords:} #1}

\smallskip
\setlength{\hangindent}{\parindent}

\maketitle

\section{Introduction}

Chemical reaction networks, where a finite number of molecular species interact with each other through a fixed number of reaction channels, are an important tool
for modeling many biochemical systems.
The reaction dynamics is invariably \textit{noisy} due to the discrete nature of molecular interactions which causes
the timing of reactions to be random.
It is known that this noise can be neglected for systems with high copy-numbers of all the species,
and the dynamics can be modeled deterministically through a set of ordinary differential equations \cite{goutsias2007classical}.
However many biological systems involve molecular species with low copy-numbers and hence the randomness in the dynamics can have a significant impact on the properties of the system. This has been demonstrated in a number of biological systems,
such as gene expression \cite{mcadams1997stochastic}, piliation of bacteria \cite{munsky2005stochastic} and polarization of cells \cite{fange2006noise},
as well as in synthetic biological circuits such as the Toggle Switch \cite{Gardner2000} and the Repressilator \cite{elowitz2000synthetic}.

To understand the role of noise and its effects, reaction networks are commonly modeled as stochastic processes
with Markovian dynamics where the state represents the copy-numbers of the molecular species \cite{goutsias2013markovian}.
\trred{
The dynamics of the distribution of a Markov process representing the reaction network evolves according to the Chemical Master Equation (CME) which is a set of ordinary differential
equations \eqref{eq:chemical_master_equation}.
}
The size of this system is equal to the number of elements in the state space, which is typically infinite,
making the task of solving the CME practically impossible for most interesting systems.
One can solve a projection of the CME onto a finite subspace, but this only works for small systems \cite{munsky2006finite}.
Other methods are developed that allow the direct solution of the CME for specific classes of networks \cite{Kazeev2014}.

If the CME cannot be solved directly, one usually resorts to Monte Carlo methods
to generate trajectories of the underlying stochastic process
and approximate the probability
distribution through a large number of simulations.
These simulations can be performed using Gillespie's Stochastic Simulation Algorithm (SSA)
or its variants \cite{gillespie2007stochastic},
that generate exact sample paths by taking into account the firing of each reaction within the simulation time-period.
Several biological systems have reactions spanning a wide range of timescales, either due to
variation in the magnitude of the rate constants or variations in the copy-numbers of the constituent species.
If the system has reactions with fast timescales, then simulation schemes like SSA can take an impractical amount of time \cite{E2005, Cao2005}.

To handle this problem, approximate schemes such as $\tau$-leaping methods have been developed, that perform multiple reactions at each step \cite{gillespie2001approximate}.
Such methods are very efficient in simulating many systems with
high reaction rates and many enhancements have been proposed to automatically select a proper leap-size \cite{cao2006efficient,cao2007adaptive}
or deal with networks where the dynamics is ``stiff'' and some reactions occur on a very fast timescale \cite{rathinam2003stiffness}.

In most biological applications, one is interested in computing the probability distributions
described by a CME.
If we can capture the necessary stochasticity using approximate sample paths that are easier to simulate,
then  we can efficiently obtain a close approximation to the solution of a CME.
For this to work we need to identify and conserve the important sources of stochasticity while discarding the insignificant
stochastic effects.
The dynamical law of large numbers suggests that stochastic effects are less important for a species if the copy-numbers are large \cite{kurtz1978strong}.
For example, consider a simple gene-expression network, where DNA is transcribed into mRNA and mRNA is translated
into proteins. One might expect that the stochasticity of transcription and
translation of low copy-number DNA and mRNA can be important but the stochasticity of the high-copy
number proteins for downstream processes to be unimportant.
If the copy-numbers of \textit{all} the species are large, then under a suitable scaling of rate constants,
the process describing the
species concentrations \footnote{The concentration is the copy-number divided by the volume of the system},
converges in the limit of infinite copy-numbers to the reaction rate equations
that correspond to the deterministic model of the reaction network \cite{kurtz1978strong}.
Instead of scaling the copy-numbers of all the species uniformly and taking the limit, one can
construct processes where the species are partitioned into a set described by concentrations
and a set described by copy-numbers. Such a process has been shown to converge in the infinite copy-number limit
for a number of examples \cite{kang2013separation,crudu2009hybrid}.
The limiting processes in these cases are hybrid processes, that combine both deterministic and Markovian dynamics
and, can be called Piecewise Deterministic Markov processes (PDMP) \cite{davis1984piecewise}.

This idea has been exploited in various
hybrid schemes \cite{haseltine2002approximate,Pahle2002,Neogi2004,Bentele2005,salis2005accurate,crudu2012convergence,alfonsi2005,
Hoops2006,Griffith2006,Puchaka2004a,Burrage2004a,Harris2006}, which mostly differ in the way that the partitioning is performed and the way in which the hybrid system is simulated.
The partitioning of the species and/or reactions is either done manually or certain threshold-values are chosen for different properties of the system,
e.g.\ the copy-numbers of species \cite{Neogi2004,Pahle2002,Hoops2006}, the propensities of reactions \cite{haseltine2002approximate,Harris2006}, the approximated copy-number fluctuations of species \cite{Bentele2005} or a combination of species copy-numbers and reaction propensities \cite{salis2005accurate,alfonsi2005,Griffith2006,Burrage2004a,Puchaka2004a}.
For the simulation, various combinations of SSA, $\tau$-leaping, stochastic differential equations (SDE) and ordinary differential equations (ODE) are possible.
Most methods combine the SSA approach with a SDE or an ODE approach \cite{haseltine2002approximate,salis2005accurate,Griffith2006,alfonsi2005,Pahle2002,Hoops2006} but there
are also methods that introduce a regime for $\tau$-leaping \cite{Puchaka2004a,Burrage2004a,Harris2006}.
Another important point for the simulation is the way in which the timings of discrete events is approximated. When integrating
the system with $\tau$-leaping, SDE or ODE approach, the propensities of the discrete reactions change with time.
Incorporating these time-dependent propensities into the sampling of the discrete event times improves the accuracy of these methods.
As soon as the criteria for the partitioning of the system is defined, one can establish an adaptive scheme to repartition the system when necessary.
This is important for situations where the orders of magnitudes of the species copy-numbers change significantly over time,
a feature found in many important reaction networks from Systems Biology such as
transcriptional bursting in gene expression \cite{chubb2006transcriptional} or the synthetic Repressilator circuit \cite{elowitz2000synthetic}.
How to properly decide when a repartitioning is necessary remains an open issue.
We refer the reader to \cite{pahle2009biochemical_b} for a more comprehensive overview.

The dynamical law of large numbers can be exploited to approximate the dynamics
when reaction timescales differ due to variations in the copy-numbers of
the constituent species. As mentioned before, timescale separation between
reactions can also be caused due to differences in the magnitudes of the rate constants \cite{Cao2005}.
Some hybrid schemes simplify the dynamics in this situation by applying the quasi-stationary assumption
for the dynamics of fast subnetworks \cite{crudu2012convergence}.
This approach is justified in situations where the exact dynamical details of fast
subnetworks are unimportant, but only the distributions of the fast species
are required to correctly estimate the propensities of the slow reactions.
For example, in many biochemical systems the quasi-stationary assumption can
be applied to fast subnetworks consisting of enzyme-substrate interactions
and exact simulations of the fast dynamics can be avoided, significantly reducing
the computational effort \cite{Rao2003,Cao2005}.

In this paper we propose a new hybrid scheme to estimate solutions of the CME
corresponding to a stochastic reaction network, which may exhibit multiple
reaction timescales due to variation in both the magnitudes of rate constants and copy-numbers of the constituent species.
Our method relies on the rigorous mathematical framework recently provided
by Kang et al.\ \cite{kang2013separation} to simplify the dynamics for such networks.
In this framework they introduce scaling parameters for rate constants
as well as species copy-numbers and give formal criteria for convergence
to a PDMP.
Using Linear Programming we select these scaling parameters to fulfil their
criteria and then simulate the limiting PDMP.
When the magnitudes of species copy-numbers change significantly, our method
adapts itself by recomputing the appropriate scaling parameters and
continues the simulation with the corresponding PDMP.
In other words, our method dynamically stitches together several PDMPs to account for the variations
in the magnitudes of the species copy-numbers, and the resulting variations in the
timescales of various reactions.
The method can also automatically identify fast subnetworks and apply
the correct quasi-stationary assumption whenever possible.

The paper is organized as follows. In Section \ref{sec:MATHEMATICAL-OVERVIEW}
we give a mathematical background and introduce the necessary formalism
for the rest of the paper. In Section \ref{sec:IMPLEMENTATION} we
elaborate on the algorithm and provide implementation details
of our adaptive hybrid scheme.
In Section \ref{sec:COMBINATION_TAU_LEAP} we suggest a possible combination of our scheme with
$\tau$-leaping schemes.
\trred{In Section \ref{sec:NUMERICAL-EXAMPLES} we compare
our adaptive hybrid scheme with SSA and a fixed PDMP scheme (without dynamic repartitioning) for
three different examples. We also make a comparison with an existing hybrid scheme (which also dynamically repartitions).
Finally, in Section \ref{sec:CONCLUSIONS} we conclude and give an outlook for future work.
Additionally, in Appendix \ref{appendix:averaging} we explain the application of the quasi-stationary assumption and
in Appendix \ref{appendix:MATHEMATICAL_JUSTIFICATION} we provide a mathematical justification
for the correctness of our adaptive hybrid scheme.
}

\section{Mathematical Preliminaries \label{sec:MATHEMATICAL-OVERVIEW}}

\trred{
The goal of this section is to present the relevant mathematical background for this paper. We start by describing how the dynamics of a reaction network can be modeled as a continuous time Markov chain. Such models are called \emph{Stochastic Reaction Networks} (SRNs) in this paper. We then motivate and describe the multiscale modeling framework of Kang et al.\ \cite{kang2013separation}, which explicitly accounts for the variation in reaction timescales by considering both copy-number scales and the differences in magnitudes of rate constants. We also present the limit theorem proved in Kang et al.\ \cite{kang2013separation} which shows that under suitable assumptions, the dynamics of the underlying stochastic process is well-approximated by a \emph{Piecewise Deterministic Markov Process} (PDMP) which can be far easier to simulate than the original dynamics. Finally we end this section with a brief discussion on how the simplified dynamics can be obtained even in situations where the required assumptions for the limit theorem fail. In such cases it is often possible to first apply the quasi-stationary assumption
and then use the PDMP approximation as before.
}

\subsection*{Stochastic Reaction Network (SRN)}

\trred{
Consider a well-stirred chemical system with $n_{S}$ species
$\left\{ s_{1},\ldots,s_{n_{S}}\right\}$ which interact according to $n_{R}$ reaction channels
of the form
\[
\sum_{i=1}^{n_{S}}\nu_{ik}s_{i} \longrightarrow \sum_{i=1}^{n_{S}}\nu'_{ik}s_{i}, \quad k\in\left\{ 1,\ldots,n_{R}\right\}
\]
where $\nu_{ik}$ and $\nu'_{ik}$ denote the number of molecules of the $i$-th species that are consumed and produced by the $k$-th reaction. In the stochastic setting, the reaction dynamics can be represented as a continuous time Markov chain, whose state at any time $t$ is just the vector $x = \left(x_1,\ldots,x_{n_S}\right)$ of copy-numbers of all the species.
When the state is $x$, the $k$-th reaction occurs after a random time
that is exponentially distributed with rate $\lambda'_{k}\left(x\right)$, assuming no other reactions occur first. The function $\lambda'_{k}$ is called the \emph{propensity function} for the $k$-th reaction. We assume that the network only consists of elementary reactions\footnote{
IUPAC. Compendium of Chemical Terminology, 2nd ed. (the \textit{Gold Book}). Compiled by A. D. McNaught and A. Wilkinson. Blackwell Scientific Publications, Oxford (1997). XML on-line corrected version: http://goldbook.iupac.org
(2006-) created by M. Nic, J. Jirat, B. Kosata; updates compiled by A. Jenkins. ISBN 0-9678550-9-8. doi:10.1351/goldbook.
}
and the propensity functions follow mass action kinetics \cite{gillespie2009diffusional}, as shown in Table \ref{tab:mass_action_table}.
}
\begin{table}[h]
\begin{center}
\begin{tabular}{|c|c|c|}
\hline
Reaction type & Consumption stoichiometry & Propensity function \tabularnewline
\hline
Constitutive & $\nu_{k} = 0$ & $\lambda'_{k}\left(x\right) = \kappa'_{k}$ \tabularnewline
Monomolecular & $\nu_{k} = e_{i}$ & $\lambda'_{k}\left(x\right) = \kappa'_{k} x_{i}$ \tabularnewline
Bimolecular & $\nu_{k} = e_{i} + e_{j}$, $i \neq j$ & $\lambda'_{k}\left(x\right) = \kappa'_{k} x_{i} x_{j}$ \tabularnewline
Bimolecular & $\nu_{k} = 2 e_{i}$ & $\lambda'_{k}\left(x\right) = \frac{1}{2} \kappa'_{k} x_{i} (x_{i} - 1)$ \tabularnewline
\hline
\end{tabular}
\end{center}
\caption{Description of mass action kinetics. Here $e_{i}$ is the $i$-th standard basis vector in $\mathbb{R}^{n_{S}}$.}
\label{tab:mass_action_table}
\end{table}

\trred{
The CME corresponding to this SRN is given by
\begin{equation}
\label{eq:chemical_master_equation}
\frac{dp\left(x,t\right)}{dt}=\sum_{k=1}^{n_{R}}\lambda'_{k}\left(x-\xi_{k}\right)p\left(x-\xi_{k},t\right)-\sum_{k=1}^{n_{R}}\lambda'_{k}\left(x\right)p\left(x,t\right)
\end{equation}
where $\xi_{k}=\nu'_{k}-\nu_{k}$, $\nu_k = (\nu_{1k} ,\dots, \nu_{ n_s k}  )$ and $\nu'_k = (\nu'_{1k} ,\dots, \nu'_{ n_s k}  )$.
In this equation, $p\left(x,t\right)$ is the probability $\mathbb{P}\left(X\left(t\right) = x\right)$, where 
$\left\{X\left(t\right): t \geq 0\right\}$ is the Markov process describing the reaction dynamics. This process satisfies the following  random time-change representation \cite{ethier2009markov}
\begin{align}
\label{eq:random_time_SRN}
\begin{split}
X\left(t\right) = x(0) + \sum_{k=1}^{n_{R}} Y_{k}\left(\int_{0}^{t}\lambda'_{k}\left(X(s)\right)ds\right) 
  \xi_{k}
\end{split}
\end{align}
where $\left\{Y_{k}: k=1,\ldots,n_{R}\right\}$ is a family of independent unit rate Poisson processes.
}

\trred{
In order to understand the behavior of a SRN it is important to compute the distributions that evolve according to the CME. However note that the CME consists of as many equations as the number of elements in the state space of the reaction dynamics, which is typically very large or infinite. This makes CME practically impossible to solve in most cases and therefore its solution is usually estimated using Monte Carlo methods such as
Gillespie's SSA \cite{gillespie2007stochastic} which generates exact paths of the underlying stochastic dynamics. As mentioned before, many biological networks contain reactions firing in several different timescales, and in such situations these exact Monte Carlo methods can be computationally intractable due to the ``stiffness'' of the system \cite{rathinam2003stiffness}. In these cases, it is sometimes possible to approximate the solutions of a CME using a simplified description of the dynamics, which is obtained by preserving the important sources of stochasticity in the dynamics while discarding the rest. One way to obtain such simplified dynamics is to use the multiscale modeling framework of Kang et al.\ \cite{kang2013separation} and show that under certain conditions, the orginal dynamics is well-approximated by a PDMP \cite{davis1984piecewise} which is often much easier to simulate. We now describe this multiscale modeling framework.
}

\subsection*{Multiscale models of Stochastic Reaction Networks}
\label{sec:multiscale_framework}

\trred{
A multiscale SRN is characterized by reactions firing at many different timescales. This variation in timescales could be both due to variation in copy-number scales as well as variation in the magnitudes of the rate constants. Kang et al.\ \cite{kang2013separation} present a framework for separating both these sources of timescale variation, and obtain PDMP approximations of the dynamics under certain conditions. To apply this framework we first select a large positive number $N_{0}$ that reflects the typical copy-number of a species which is considered abundant in the network. Note that the PDMP approximation that we later describe would ignore the stochasticity in copy-numbers of species if they are of order $N_0^{\alpha}$ for some $\alpha >0$. Hence the accuracy of this approximation and our method which relies on it, crucially depends on the choice of $N_{0}$.
}

\trred{
As an example, consider a gene expression network where the number of genes
is usually $1$ or $2$, the number of mRNAs is in the order of tens and the number of proteins is in the order of thousands.
To capture the bursting behavior of stochastic gene expression \cite{chubb2006transcriptional} we want to keep the
discreteness of the genes and the mRNAs.
However, the discreteness of reaction channels driven by proteins can often be neglected, as the fluctuations of the protein numbers
are rather small compared to the total protein number. Thus we want to approximate the number of proteins as a continuous
quantity and reaction channels driven by proteins as occurring continuously. For this purpose we can choose $N_{0} = 1000$.
}

\trred{
From now on, we refer to a quantity as order $1$ (denoted as $O(1)$) if it remains bounded as $N_0$ gets larger. Assume that the rate constants $\kappa'_{k}$ have different orders of magnitude. For each reaction $k=1,\ldots,n_{R}$ we pick a $\beta_{k} \in \mathbb{R}$, such that $\kappa_{k} = N_{0}^{-\beta_{k}} \kappa'_{k}$ is $O(1)$. We can interpret $\beta_{k}$ as the timescale of the $k$-th reaction, when all the reactants have copy-numbers of order one. 
To account for the variation in the copy-number scales of species we introduce
scaling parameters $\alpha_{i} \geq 0$ for $i=1,\ldots,n_{S}$ such that
$Z_{i}^{N_{0}} \left(t\right) = N_{0}^{-\alpha_{i}} X_{i}\left(t	\right)$
is $O(1)$, at least for $t$ close to $0$. The choice of scaling parameters $\alpha_{i}$-s and $\beta_{k}$-s is not arbitrary as they must satisfy certain conditions that we later describe. In Section \ref{subsec:computation_scaling_parameters} we will present an automatic scheme for selecting these parameters in a suitable way.
}

\trred{
Using these scaling parameters $\alpha_{i}$-s and $\beta_{k}$-s, we can derive the scaled process $\{ Z^{N_{0}}\left(t\right) = (Z_{1}^{N_{0}}\left(t\right),\ldots,Z_{n_{S}}^{N_{0}} (t) ) : t \geq 0 \}$ from the original process $\{ X\left(t\right) : t \geq 0 \}$.
Replacing $N_{0}$ by $N$ we obtain a family of processes $\left\{Z^{N}\left(t\right) : t \geq 0\right\}$
parametrized by $N$. If we can show that the process $Z^{N}$ converges in distribution to another process $Z$ as $N\rightarrow\infty$, then for large values of $N_{0}$, $Z_{i}^{N_{0}}\left(t\right) \approx Z_{i}\left(t\right)$
and hence $X_{i}\left(t\right)\approx N_{0}^{\alpha_{i}} Z_{i}\left(t\right)$ for all $i$.
This gives us a way to approximate our original process $X$ with another process
$Z$ which can be much simpler to simulate. In the classical thermodynamic limit \cite{kurtz1978strong},
$N_{0}$ is the volume of the system, $\alpha_{i} = 1$ for each $i$
and $\beta_{k} = 1, 0, -1$ depending on whether the $k$-th reaction is constitutive, monomolecular or bimolecular.
In this case the limiting process $Z$ corresponds to the deterministic model of the reaction network and its evolution is given by a systems of ODEs called the reaction rate equations. However for many biochemical networks this classical scaling is not suitable, but a general scaling prescribed by $\left(\alpha_{i},\beta_{k}\right)$ can be used. We now present the main convergence result on which our method is based.
}

\subsection*{Convergence to a PDMP}
\label{subsec:PDMP}

\trred{
Under the scaling prescribed by $\left(\alpha_{i},\beta_{k}\right)$, the random time-change representation of the process
$\{ Z^{N} (t) = (Z_{1}^{N}(t),\ldots,Z_{n_{S}}^{N} (t) ) : t \geq 0 \}$ is given by
\begin{align}
\label{eq:multiscale_random_time_change}
Z_{i}^{N}(t) = Z_{i}^{N}(0) + N^{-\alpha_{i}}\sum_{k=1}^{n_R}Y_{k}\left(\int_{0}^{t}N^{\beta_{k}+\alpha\cdot\nu_{k}}\lambda_{k}^{N}\left(Z^{N}(s)\right)ds\right)\xi_{ik}
\end{align}
where $Z_i^N (0) = N_{0}^{-\alpha_i} X_i (0)$ and $\alpha\cdot\nu_{k}$ denotes the dot product of vectors $\alpha = \left(\alpha_{1},\ldots,\alpha_{n_{S}}\right)$
and $\nu_{k}$.
The $\lambda_{k}^{N}$-s are defined by
}
\begin{center}
\begin{tabular}{|c||c|c|c|c|}
\hline
Reaction & $\emptyset\rightarrow\cdots$ & $S_{i}\rightarrow\cdots$ & $S_{i}+S_{j}\rightarrow\cdots$ & $2S_{i}\rightarrow\cdots$\tabularnewline
\hline 
$\lambda_{k}^{N}\left(z\right)$ & $\kappa_{k}$ & $\kappa_{k}z_{i}$ & $\kappa_{k}z_{i}z_{j}$ & $\kappa_{k}z_{i}(z_{i}-N^{-\alpha_{i}})$\tabularnewline
\hline 
\end{tabular}
\end{center}
\trred{
with $\kappa_{k}=N_0^{-\beta_{k}}\kappa'_{k}$.  In Theorem 4.1.\ Kang et al.\ \cite{kang2013separation}\ prove that the process $Z^{N}$ converges to a well-behaved process $Z$ as $N\rightarrow\infty$, if $\left(\alpha_{i},\beta_{k}\right)$ satisfy
\begin{align}
\label{eq:multiscale_constraints}
\alpha_{i} \geq \beta_{k} + \alpha\cdot\nu_{k} \qquad \text{ for each } i, k \text{ with } \xi_{ik} \neq 0.
\end{align}
Moreover the limiting process $Z$ satisfies
\begin{align}
\label{eq:random_time_PDMP}
Z\left(t\right) = x(0)
 + \sum_{k\in R_{D}}Y_{k}\left(\int_{0}^{t}\lambda_{k}\left(Z(s)\right)ds\right) \xi_{k}
 + \sum_{k\in R_{C}} \left( \int_{0}^{t}\lambda_{k}\left(Z(s)\right)ds \right)
  \xi_{k},
\end{align}
where $\left\{Y_{k}:  k \in R_{D} \right\}$ is a family of independent unit rate Poisson processes, and reaction subsets $R_{C}  , R_{D} \subset \left\{1,\ldots,n_{R}\right\}$ are defined by
\begin{align}
\label{eq:multiscale_partitioning}
R_{C} = & \left\{ k: \alpha_{i} = \beta_{k} + \alpha \cdot \nu_{k} > 0 \ \text{ for all } i \text{ such that } \xi_{ik} \neq 0 \right\}
\text{ and} \notag \\
R_{D} = & \left\{ k \not \in R_{C}: \alpha_{i} = \beta_{k} + \alpha \cdot \nu_{k} = 0 \ \text{ for all } i \text{ such that } \xi_{ik} \neq 0 \right\}.
\end{align}
}

\trred{
The limiting process $Z$ described by \eqref{eq:random_time_PDMP} is essentially a \emph{Piecewise Determinitic Markov Process} (PDMP) because a part of the dynamics, given by reactions in $R_C$ evolves continuously (as in an ODE), while another part of the dynamics,  given by reactions in $R_D$ evolves discretely like a jump Markov process similar to our original process $X$. Note that in general $R_{C} \cup R_{D} \neq \left\{1,\ldots,n_{R}\right\}$ and all reactions not in $R_{C} \cup R_{D}$ do not contribute to the dynamics of the limiting process $Z$.  
}

\trred{
With the scaling parameters $\left(\alpha_{i},\beta_{k}\right)$ we also define subsets of species by
\begin{align}
\label{eq:multiscale_species_partitioning}
S_{C} = & \left\{ i \in \left\{1,\ldots,n_{S}\right\}: \alpha_{i} > 0 \right\} \notag \\
S_{D} = & \left\{1,\ldots,n_{S}\right\} \setminus S_{C} .
\end{align}
The species in $S_{C}$ change due to reactions in $R_{C}$ and therefore they are measured by ``concentrations'' and evolve continuously,
whereas the species in $S_{D}$ will have discrete state values given by their copy-number.
The dynamics of the limiting process $Z$ can then be written as
\begin{align}
\label{eqn:dynamics_partitioning}
Z_{C}\left(t\right) & = x_{C}(0)
 + \sum_{k\in R_{C}} \left( \int_{0}^{t}\lambda_{k}\left(Z(s)\right)ds \right) \xi_{k}
  \notag \\
Z_{D}\left(t\right) & = x_{D}(0)
 + \sum_{k\in R_{D}}Y_{k}\left(\int_{0}^{t}\lambda_{k}\left(Z(s)\right)ds\right) 
  \xi_{k}
\end{align}
where $Z_{C} \left(t\right)$ and $Z_{D} \left(t\right)$ are vectors with entries $Z_{i}\left(t\right)$ for $i\in S_{C}$ and $i \in S_{D}$ respectively.
The definitions of $x_{C}(0)$ and $x_{D}(0)$ are similar.
}

\trred{
To simulate such a PDMP, the dynamics of the reactions in $R_{C}$ are basically an \emph{Ordinary Differential Equation} (ODE) and the reactions in $R_{D}$ can modify its vector field at random times.
The ODE can be solved with established and efficient ODE solvers and for large propensities of the reactions channels in $R_{C}$, this can be computationally much cheaper than simulating
the discrete individual reactions as Gillespie's SSA \cite{gillespie2007stochastic}.
Therefore the speed enhancements that our method obtains (by simulating a PDMP \eqref{eq:random_time_PDMP}) in comparison to SSA, is directly
proportional to the number of reactions treated continuously.
More details on simulating a PDMP are provided in section \ref{sec:IMPLEMENTATION}.
}

\trred{
For the above convergence result to hold, we need to find scaling parameters $\left(\alpha_{i}, \beta_{k}\right)$ such that \eqref{eq:multiscale_constraints} is satisfied.
Note that such a choice always exists as we can simply set each $\alpha_i$ and $\beta_{k}$ to $0$ to satisfy these contraints. However with such a trivial choice of scaling parameters, the limiting process $Z$ is same as the original process $X$ and hence we do not obtain any computational advantage in simulating the limiting process. In Section \ref{subsec:computation_scaling_parameters} we describe how these parameters can be automatically chosen using Linear Programming in such a way, that the limiting dynamics is computationally much easier to simulate.
The key idea is to treat as many reactions continuously as possible without violating the constraints that imply the PDMP convergence result.
We also discuss how they can be properly adapted if the copy-number scales of species vary significantly with time.
}

 \subsection*{Applying the quasi-stationary assumption}
\label{subsec:qsa}

\trred{
Consider the situation where certain species have low copy-numbers but their dynamics is affected by reactions with large rate constants $\kappa'_{k}$. In this situation, an appropriate copy-number scaling parameter $\alpha_{i}$ for such species would be close to $0$, since copy-numbers are small, while an appropriate rate constant scaling parameter $\beta_{k}$ for these reactions would be strictly positive.
As a consequence, for these low-copy species with \emph{fast} dynamics (or simply \emph{fast} species), the constraits \eqref{eq:multiscale_constraints} will fail to hold for scaling parameters
$\left(\alpha_{i}, \beta_{k}\right)$ that truly represent the dynamics. If we use our automatic parameter selection procedure given in Section \ref{subsec:computation_scaling_parameters}, then all the reactions involved in changing these \emph{fast} species will have their $\beta_{k}$-s set to a negative or a small positive value even though their rate constants $\kappa'_{k}$-s are large. Hence our method will not be able to treat these fast reactions continuously and therefore it will suffer from the same stiffness problems as Gillespie's SSA \cite{Cao2005}.
}

\trred{
To remedy this problem we use another result in Kang et al.\ \cite{kang2013separation}, which shows that even with such \emph{fast} species in the network,
well-behaved limits can be obtained for certain projections of the process $Z^{N}$.
Such projections correspond to ``reduced'' reaction networks
which only consist of those species that satisfy \eqref{eq:multiscale_constraints}.
These reduced models can also consist of linear combinations of \emph{fast} species that satisfy a constraint analogous to \eqref{eq:multiscale_constraints}. To obtain these reduced models along with their limits,
one needs to find subnetworks of fast reaction channels where the quasi-stationary
assumption can be applied.
Assume that we have such a subnetwork whose internal dynamics is much faster than the surrounding dynamics.
If the subnetwork dynamics converges to a stationary distribution
then we can use this distribution to estimate the propensities of
reactions emanating from this subnetwork \cite{pelevs2006reduction, radulescu2008robust,crudu2012convergence}.
In this approach we assume that the species influenced by these fast reaction channels
are always at stationarity and hence we do not capture their transient dynamics,
thereby saving a lot of computational time \cite{Rao2003,Cao2005}.
The application of quasi-stationary approximation is also called
averaging in the mathematical literature.
We discuss this approach in greater detail in Appendix \ref{appendix:averaging}.
}

\color{black}
\section{Implementation\label{sec:IMPLEMENTATION}}

The aim of this section is to provide full implementation details
of our method for approximating the solution of a CME.
The presented algorithms compute a single sample path of the dynamics. The solution of a CME at a certain time $t$ can be approximated by generating
several sample paths and computing the histogram of the sample-path values at time $t$.
We start by introducing the well-known scheme for simulating a PDMP,
where the partitioning of reactions/species into discrete/continuous subsets is fixed.
We then describe our adaptive hybrid scheme, where this partitioning changes with time
due to the variation in the orders of magnitude of species copy-numbers.
This partitioning is based on the scaling parameters $\alpha_{i}, \beta_{k}$,
mentioned before, and we describe
how these parameters can be appropriately chosen using Linear Programming.
In both these approaches (Fixed and Adaptive PDMP) computational time can be saved
by applying the quasi-stationary assumption on fast subnetworks
consisting of discrete species and discrete reactions.
We explain this procedure and indicate how it can be integrated into our method.

\trred{The inputs that are the same to all the algorithms are given by:
\begin{description}
	\setlength{\itemsep}{0pt}
	\item[$\mathbf{t_0}$] initial time of simulation
	\item[$\mathbf{t_f}$] final time of simulation
	\item[$\mathbf{z_0}$] initial state vector of simulation
	\item[$\nu_{ik}$ for $i\in\left\{1,\ldots,n_{S}\right\}$ and $k\in\left\{1,\ldots,n_{R}\right\}$] the consumption stoichiometries
	\item[$\nu'_{ik}$ for $i\in\left\{1,\ldots,n_{S}\right\}$ and $k\in\left\{1,\ldots,n_{R}\right\}$] the production stoichiometries
	\item[$\kappa'_{k}$ for $k\in\left\{1,\ldots,n_{R}\right\}$] the unscaled reaction rate constants
	\item[$N_{0} \in \mathbb{N}$] copy-number scale considered to be large (default value $N_{0} = 1000$)
\end{description}
The adaptive algorithm needs two additional inputs:
\begin{description}
	\item[$\mu$] continuous copy-number scale threshold (default value $\mu$ = 0.5)
	\item[$\eta$] adaptation scale threshold (default value $\eta = 0.9$)
\end{description}
These are described in more detail in the following sections.
}

We will use the following notation throughout this section:
\begin{description}
	\setlength{\itemsep}{0pt}
	\item[$\mathbf{P}$] total number of firings considering all discrete reactions occurring from $t_{0}$ until $t_{f}$
	\item[$\mathbf{t_{r} (p)}$] time of the $p$-th firing considering all discrete reactions for $p=1,\ldots,P$.
		We define $t_{r}\left(0\right)=t_{0}$ and $t_{r}\left(P+1\right) = \infty$.
	\item[$\mathbf{R_{D} (t})$] the set of discrete reaction channels at time $t$
	\item[$\mathbf{R_{C} (t)}$] the set of continuous reaction channels at time $t$
	\item[$\mathbf{S_{D} (t})$] the set of discrete species at time $t$
	\item[$\mathbf{S_{C} (t)}$] the set of continuous species at time $t$
	\item[$\mathbf{z\left(t\right)}$] the state vector at time $t$
\end{description}
Note that the entries in the state vector $z\left(t\right)$ corresponding to
$S_{D}\left(t\right)$ are non-negative integers representing the species copy-numbers
whereas the entries corresponding to
$S_{C}\left(t\right)$ are non-negative real numbers representing the scaled species copy-numbers.

\subsection*{Simulation of Fixed Piecewise Deterministic Markov Processes}

\trred{
Consider a typical PDMP model where the continuous and discrete reaction subsets, given by $R_{C}\left(t\right)$ and $R_{D}\left(t\right)$ respectively,
and the induced species partition $S_{C}\left(t\right)$ and $S_{D}\left(t\right)$ do not depend on the time $t$.
}
The usual algorithm to simulate such models is to evolve the part
of the model that is described deterministically until the next discrete
reaction occurs. When a discrete reaction occurs, the copy-numbers are updated accordingly.
This is repeated until the end-timepoint of the simulation
is reached. To evolve the deterministic system we use existing ODE solvers.
\trred{\footnote{\trred{The implementation allows different ODE solvers to be used. However, for the examples we used the Dormand-Prince 5 stepper from the boost odeint library (\url{http://www.boost.org/}).
The boost odeint library is coupled to our implementation via a Java Native Interface wrapper}}}

The time increment $\Delta t_{p} = t_{r}\left(p\right) - t_{r}\left(p - 1\right)$
between the $(p - 1)$-st and $p$-st reaction is defined by the stopping condition
\begin{align}
\label{eq:PDMP_stopping_condition}
\int_{t_{r}\left(p-1\right)}^{t_{r}\left(p-1\right) + \Delta t_{p}} \sum_{k\in R_{D}}\lambda_{k}\left(z\left(t\right)\right) = u_{p}
\end{align}
where $u_{p} = - \log\left(q_{p}\right)$, $q_{p} \sim \mathcal{U}\left[0,1\right]$ and $\mathcal{U}\left[0,1\right]$ is the uniform distribution over $[0,1]$.
At time $t_{r}\left(p\right)$ the $k$-th reaction in $R_{D}$ fires with probability
\begin{align}
\label{eq:PDMP_reaction_probability}
\frac{\lambda_{k}\left(z\left(t_{r}\left(p\right)\right)\right)}
	{\sum_{m\in R_{D}} \lambda_{m}\left(z\left(t_{r}\left(p\right)\right)\right)} \quad .
\end{align}
For time $t_{r}\left(p - 1\right) < t < t_{r}\left(p\right)$, only the states of continuous species will change
according to the ordinary differential equation given by
\begin{align}
\label{eq:PDMP_vector_field}
\frac{d}{dt}z\left(t\right) = \sum\limits_{k\in R_{C}} \lambda_k\left(z\left(t\right)\right) \xi_{k} \quad .
\end{align}

We start the simulation at time $t = t_{0}$ and $z\left(t_{0}\right) = z_{0}$ and
generate a random variable $u_{1}$ with distribution $\mathcal{U}\left[0,1\right]$.
We evolve the state $z\left(t\right)$ according to \eqref{eq:PDMP_vector_field} until
the first discrete reaction occurs at
time $t_{r}\left(1\right) = t_{0} + \Delta t_{1}$, defined by \eqref{eq:PDMP_stopping_condition}.
We sample the reaction $k$ from $R_{D}$ according to \eqref{eq:PDMP_reaction_probability} and
set $z\left(t_{r}\left(1\right)\right) = z\left(t_{r}\left(1\right) - \right) + \xi_{k}$.
These steps are repeated until we reach the final time $t_{f}$.
Below we give an algorithmic description of the procedure we just described.
\medskip

\textbf{Algorithm 1: Fixed PDMP}
\begin{my_algorithm}
	\item $t\gets t_0$, $z\left(t_0 \right)\gets z_0$, $p\gets 1$
 	\item \textbf{while} \trred{$t < t_{f}$} \textbf{do}
		\item \ \% Set up random variable for the stopping condition \eqref{eq:PDMP_stopping_condition}
		\item \ $u_{p} \gets - \log\left( q_{p} \right)$ where $q_{p} \sim \mathcal{U}\left[0,1\right]$
		\item \ \textbf{evolve}
			\item \ \ $z$ according to $\frac{d}{dt}z\left(t\right) = \sum\limits_{k\in R_{C}} \lambda_k\left(z\left(t\right)\right) \xi_{k}$
		\item \ \textbf{until} \trred{$t=t_{f}$} or $t = t_{r}\left(p\right)$ according to \eqref{eq:PDMP_stopping_condition}
		\item \ \textbf{if} $t = t_{r}\left(p\right)$ \textbf{then}
			\item \ \ \% Determine which discrete reaction occurs at $t$ and \\update state vector
			\item \ \ Sample $r$ with $\text{Pr}\left(r = k\right) \sim \lambda_k\left(z\left(t\right)\right), k \in R_{D}$
			\item \ \ $z\left(t\right) \gets z\left(t-\right) + \xi_r$
			\item \ \ $p \gets p + 1$
		\item \ \textbf{end if}
	\item \textbf{end while}
\end{my_algorithm}

\subsection*{Adaptive Simulation of Piecewise Deterministic Markov Processes}

For SRNs that show big variation in copy-numbers over
time, it may be impossible to find a fixed partitioning of the
reaction and species sets that works well until the final time $t_{f}$.
This motivated us to extend the fixed PDMP algorithm into an adaptive
scheme, where the partitioning can change with time, depending on the
copy-number scales of different species. The general idea
is to define bounds for the copy-numbers in a suitable manner and
upon leaving these bounds the partitioning of the reaction
and species sets is updated according to the current copy-number scales.

The copy-number bounds are defined by parameters $\mu\ge0$ and $\eta\ge0$.
For a discrete species the
upper bound is $N_{0}^{\mu}$, while for a continuous species the lower
and upper bounds are $N_{0}^{-\eta}$ and $N_{0}^{\eta}$ respectively.
Thus the parameter $\mu$ is used to decide when a species can be considered
continuous and the parameter $\eta$ gives the range of re-scaled
copy-numbers where no adaptation is performed.
\trred{We use $\mu = 1$ and $\eta = 0.9$ as default parameter values. Only in the case
of too few or too many adaptations do we advise to try different values for $\eta$.
Note that even though the performance of our algorithm depends on the choice of $\mu$ and $\eta$,
this dependence is relatively small.}
\trred{
For each $i\in\left\{1,\ldots,n_{S}\right\}$ define a set
\begin{equation}
\label{eq:copy_number_bounds}
B_{i}^{N_{0}} = \begin{cases}
\left[ 0, N_{0}^{\mu} \right] & \text{ if } i \in S_{D} \\
\left[ N_{0}^{-\eta},  N_{0}^{\eta} \right] & \text{ if }i\in S_{C}
\end{cases} \quad
\end{equation}
and update the partitioning of the reactions and species if the value of $x_{i}$ leaves the set $B_{i}^{N_{0}}$.
}

We implemented our adaptive simulation scheme both with a fixed time step and an adaptive time step.
Here we describe the fixed time step algorithm with a time step $dt$:
We define a subset $R_{Q} \subseteq R_{D}$ of discrete reactions
that influence the continuous dynamics
\[
R_{Q} = \left\{ k\in R_{D}: \xi_{ik} \neq 0 \text{ and } \nu_{il} \neq 0 \text{ for some } l\in R_{C} \right\} \quad .
\]
After an integration step of the continuous dynamics has been performed, the next time $t_{R}$ of a discrete reaction can be computed
by using the sum of discrete propensities $w$ computed in the subroutine \textsf{integrateStep}.
This is repeated until either the next discrete reaction time $t_{R}$ would exceed the current time of the integration $t_{N} = t + dt$, or
the discrete reaction belongs to $R_{Q}$. If the discrete reaction belongs to $R_{Q}$ then the current time of integration $t_{N}$ is set
to the discrete reaction time $t_{R}$.

If the copy-number bounds in relation \eqref{eq:copy_number_bounds} become invalid,
an adaptation procedure is performed where the scaling parameters $\alpha_{i}, \beta_{k}$
are recomputed, the state of the system is scaled accordingly and the averaging of fast subnetworks
is performed, if possible.
The adaptation and averaging procedures are described in more detail
in the following sections.

\medskip

\textbf{Algorithm 2: Adaptive PDMP}

\begin{my_algorithm}
	\item $t\gets t_0$, $z\gets z_0$, $w\gets 0$
	\item $u \gets - \log\left( q \right)$ where $q \sim \mathcal{U}\left[0,1\right]$
	\item \label{alg:outer_loop} \textbf{while} \trred{$t_0 < t_{f}$} \textbf{do}
		\item \ $t_N, [z_N, w_N]\gets \text{integrateStep(t, [z, w])}$
		\item \ $\Delta z = 0$
		\item \ \textbf{while} $w_N \geq u$ \textbf{do}
			\item \ \ $t_R\gets \text{findRoot}(t, w - U, t_N, w_N - U)$
			\item \ \ [$z_R, w_R] \gets \text{interpolate}(t_R, t, [z, w], t_N, [z_N, w_N])$
			\item \ \ $r\gets k \text{ with probability } p_k\propto \lambda_k\left(z_R + \Delta z\right)$
			\item \ \ $u\gets u - \log\left( q \right)$ where $q \sim \mathcal{U}\left[0,1\right]$
			\item \ \ \textbf{if} $r \in R_{Q}$ \textbf{then}
				\item \ \ \ $z_R\gets z_R + \xi_r$
				\item \ \ \ $t_N\gets t_R$, $z_N\gets z_R$
				\item \ \ \ Goto step \ref{alg:inner_loop_end}
			\item \ \ \textbf{else}
				\item \ \ \ $\Delta z\gets \Delta z + \xi_r$
				\item \ \ \ $t\gets t_R$, $z\gets z_R$
				\item \ \ \ $w_N\gets w_R + \sum\limits_{k\in R_{D}} \lambda_k (z_R + \Delta z) \times (t_N - t_R)$
			\item \ \ \textbf{end if}
		\item \ \textbf{end while}
		\item \label{alg:inner_loop_end} \ $t\gets t_N$, $z\gets z_N + \Delta z$, $w\gets w_N$
		\item \ \textbf{for each} $i \in S_{D} \cup S_{C}$ \textbf{do}
			\item \ \ \textbf{if} $z_i$ violates \eqref{eq:copy_number_bounds}
				\item \ \ \ Perform adaptation (Algorithm 3)
				\item \ \ \ Goto step \ref{alg:outer_loop}
			\item \ \ \textbf{end if}
		\item \ \textbf{end for}
	\item \textbf{end while}
\end{my_algorithm}

\subsection*{Computation of scaling parameters}
\label{subsec:computation_scaling_parameters}

Our method relies on computing the PDMP approximation based
on the framework by Kang et al.\cite{kang2013separation}\ described before.
For this we need to pick the scaling parameters $\alpha_{i}$-s and $\beta_{k}$-s such that
the PDMP approximation is justified.
Moreover, to enhance the computational efficiency, we would
also like to maximize the number of reactions that can be
treated continuously.
We achieve these two objectives by formulating the
conditions given by \eqref{eq:multiscale_constraints}
as constraints of a linear program that maximizes a suitably chosen weighted sum of
the $\alpha_{i}$-s and $\beta_{k}$-s.
Once these parameters have been chosen, the result of Kang et al.\cite{kang2013separation}\ 
shows that the PDMP approximation must be partitioned according to
\eqref{eq:multiscale_partitioning}.
However such a strict partitioning scheme is often inappropriate from the standpoint of numerical simulations,
because the $\alpha_{i}$-s and $\beta_{k}$-s are real numbers.
Hence we relax this partitioning criteria by introducing another parameter $\delta$ (with $\delta = 1$ as the default value)
which decides when a reaction is considered continuous or discrete.
Formally, we define
\begin{align}
\label{eq:multiscale_partitioning_implementation}
R_{C}\left(t\right) = \{k: \alpha_{i}>\delta\ \forall i \text{ s.t. } \xi_{ik}\neq 0\} \quad \text{ and } \quad
R_{D}\left(t\right) = \{k: \exists i: \alpha_{i}\leq\delta,\xi_{ik}\neq 0\}.
\end{align}
The partitioning of the species set into discrete and continuous subsets $S_{D}\left(t\right)$ and $S_{C}\left(t\right)$
is defined in Section \ref{subsec:PDMP}.

We now elaborate on the linear program that is solved to find the
scaling parameters at time $t$.
For each $i\in \left\{1,\ldots,n_{s}\right\}$ let $x_{i} = N_{0}^{\alpha_{i}} z_{i}\left(t\right)$
where the $\alpha_{i}$-s are the scaling parameters before time $t$.
For each $i\in\left\{ 1,\ldots,n_{S}\right\}$ and each $k\in\left\{ 1,\ldots,n_{R}\right\}$
define
\[
A_{i}=\frac{\log(x_{i})}{\log(N_{0})} \quad \text{ and } \quad B_{k}=\frac{\log(\kappa'_{k})}{\log(N_{0})}.
\]
Here $A_{i}$ gives the copy-number scale of the $i$-th species at time $t$,
while $B_{k}$ is a time-independent quantity
that gives the natural timescale of the $k$-th reaction, if all its reactant species have copy-numbers of order one.
We compute the new scaling parameters by solving the following linear program
\begin{center}
\begin{equation}
\begin{array}{cc}
\maximize\limits_{\alpha_i, \beta_k} & \psi \sum\limits_{i=1}^{n_{S}}\frac{\alpha_{i}}{A_{i}}+\sum\limits_{i=1}^{n_{R}}\frac{\beta_{k}}{B_{k}}\\
\\
\text{subject to } & 0\leq\alpha_{i}\leq A_{i}, \beta_{k}\leq B_{k} \quad \forall i, k\\
\text{and } & \alpha_{i}\geq\beta_{k}+\alpha\cdot\nu_{k} \quad \forall i, k \text{ with } \xi_{ik} \neq 0 \quad.
\end{array}\label{eq:adaptation_linear_program}
\end{equation}
\par\end{center}
Essentially we are trying to maximize the
values of the scaling parameters $\alpha_{i}$, $\beta_{k}$ such that
the constraints \eqref{eq:multiscale_constraints} are satisfied
and $\alpha_{i}$ approximately captures the copy-number scale of the $i$-th species
at time $t$ while $\beta_{k}$ approximately captures the natural timescale
of the $k$-th reaction.
Maximizing these scaling parameters allows us to
treat more reactions and species continuously,
and therefore enhance the speed of our simulations.
We weigh the first term in the objective function by
$\psi$ (with $\psi = 100$ as the default value) because
the correct selection of $\alpha_{i}$-s is more crucial for
the partitioning than the $\beta_{k}$-s.
We now describe how we automatically adapt the partitioning
in our method.

\medskip

\textbf{Algorithm 3: Adaptation}

\begin{my_algorithm}
	\item Solve the linear program \eqref{eq:adaptation_linear_program} to find
	the $\alpha_i$-s and $\beta_k$-s
	\item Recompute the partitions $R_{C}\left(t\right), R_{D}\left(t\right)$ and
	$S_{C}\left(t\right), S_{D}\left(t\right)$ \\according to \eqref{eq:multiscale_partitioning_implementation}
	\item Recompute copy-number bounds \eqref{eq:copy_number_bounds}
	\item Optional: Perform averaging (Algorithm 4, see next section)
\end{my_algorithm}

\subsection*{Averaging of fast subnetworks}

After the selection of the scaling parameters in the Adaptation algorithm (Algorithm 3)
we try to apply the quasi-stationary assumption on fast subnetworks, if possible.
The subnetworks that are suitable for this purpose can be precomputed at the start of the simulation procedure.
This can be done for example, by selecting subnetworks
based on conservation relationships between species, pseudo-linearity or by
checking if a subnetwork is weakly reversible with deficiency zero (see Appendix \ref{subsec:averaging_strategies}).
\trred{Also, as the number of possible subnetworks grows rapidly with the number of reactions $n_{R}$ in the network,
one has to apply some heuristic to limit the search space of subnetworks. A straightforward choice would be to only
consider subnetworks with a maximum number of reactions.}
While implementing the averaging procedure we go through the list of suitable subnetworks
that only contain discrete reactions and we compute the timescale-separation
$\Delta\zeta\left(R\right)$ as in \eqref{eq:averaging-timescale-separation} for all such subnetworks.
We select the subnetworks with $\Delta\zeta\left( R\right) \geq \Theta$ and adopt
a greedy strategy of repeatedly selecting the largest subnetwork that has not been
selected yet. This gives us a list of disjoint subnetworks on which the
quasi-stationary assumption can be applied at the current time.
Then we compute the first and second moments of the stationary distribution \footnote{
For subnetworks that are weakly reversible and have deficiency zero,
we compute the moments of the product-form stationary distribution,
where all the species satisfy a conservation relation,
by approximating it with a multinomial distribution with appropriately defined virtual species.
The error incurred by this approximation is generally small
enough to be neglected.
If no species satisfy a conservation relation, then the stationary distribution
is just the product of independent Poisson distributions.
}
for each selected subnetwork and accordingly update the rate constants of the reactions that connect the subnetwork with the surrounding network.
For a description of the implementation details see Appendix \ref{appendix:averaging}.

\section{Combination with $\tau$-leaping\label{sec:COMBINATION_TAU_LEAP}}

In this section we propose a possible extension of this work by combining $\tau$-leaping schemes with
our method to improve its speed or accuracy.
Instead of only using two regimes in our simulation, i.e.\ discrete stochastic jump dynamics and
continuous deterministic dynamics, we can add another regime in between where the
dynamics is approximated by $\tau$-leaping.
To determine which reactions are approximated by $\tau$-leaping and which ones are approximated by
continuous deterministic dynamics, the partitioning in \eqref{eq:multiscale_partitioning_implementation} can be modified
appropriately to introduce a third set of reactions $R_{\tau}$ that are approximated with $\tau$-leaping.
The step size $\tau$ of the $\tau$-leap can be computed as described in \cite{cao2006efficient}:
\[
\tau = \min_{i: I} \left\lbrace
	\frac{\max\left\lbrace \epsilon x_{i}/g_{i}, 1 \right\rbrace}{|\mu_i\left(x\right)|},
	\frac{\max\left\lbrace \epsilon x_{i}/g_{i}, 1 \right\rbrace^{2}}{|\sigma^2_i\left(x\right)|}
\right\rbrace,
\]
where
\[
	I = \left\lbrace{ i: \xi_{ik} \neq 0 \text{ with } k \in R_{\tau} \cup R_{C}}  \right\rbrace
\]
is the set of species modified by the reactions in $R_{\tau}$ or $R_{C}$.
The functions $\mu_{i}\left(x\right)$ and $\sigma_{i}^{2}\left(x\right)$ are given by
\begin{align*}
\mu_i\left(x\right) = & \sum_{k \in R_{\tau} \cup R_{C}} \xi_{ik} \lambda_{k}\left(x\right) \text{ and } \\
\sigma^2_i\left(x\right) = & \sum_{k \in R_{\tau} \cup R_{C}} \xi^2_{ik} \lambda_{k}\left(x\right) \quad , \\
\end{align*}
where the parameters $\epsilon$ and $g$ are usually taken to be $\epsilon=0.03$ and $g_{i} = 3$.
It is important to also include the continuous reactions into the step size selection as the
leap condition, i.e.\ the change of reaction propensities during the leap has to be small, might not be
satisfied otherwise.
Because of our partitioning scheme there is no need to consider reactions involving low-copy
number species in a special manner as is done in \cite{cao2006efficient} by introducing the so-called
``critical reactions''.
The integration of the continuous dynamics of the system can then be performed with a
fixed time step integration scheme where the step size is equal to $\tau$.

\section{Numerical Examples\label{sec:NUMERICAL-EXAMPLES}}

In this section we consider three examples from Systems Biology
and demonstrate that our adaptive hybrid scheme
accurately captures the solution of a CME, and it can outperform the standard SSA
and fixed PDMP schemes.
In one of the examples we also compare our scheme to an existing adaptive PDMP method by Alfonsi et al.\ \cite{alfonsi2005}.

The first example (\textit{Fast Dimerization}) illustrates the usefulness of the averaging procedure.
In this example partitioning of reactions and species is not required
and our scheme reduces to the slow-scale SSA
given in Cao et al.\cite{Cao2005}.
However, the advantage of our method is that
the quasi-stationary assumption is applied automatically
without any need for explicit computations of the stationary distributions.

The second example (\textit{Toggle Switch}) shows the performance improvements of PDMP schemes over SSA.
In this example there exists a partitioning of the reactions and the species that is valid for the whole simulation period.
Our method automatically chooses an appropriate partitioning, without the need for preceding manual analysis of the network,
unlike other existing PDMP schemes. We also show results of the method proposed by Alfonsi et al.\cite{alfonsi2005}.
\trred{Note that we implemented this method in an adaptive fashion analogous to our approach and used the same ODE solver
(i.e.\ the Dormand-Prince 5 stepper from the boost odeint library) for the continuous dynamics.}

Finally, the third example (\textit{Repressilator}) demonstrates the advantages
of using an Adaptive PDMP scheme for networks that exhibit large variations in the
copy-numbers of their species.
In this example a fixed partitioning will not work and hence
our adaptive method will outperform any fixed PDMP approach
in terms of computational efficiency and accuracy.

All simulations were run in parallel on a grid using $128$ cores
and using the parameter values for our algorithm given in Table \ref{tab:parameters}.
A comparison of the runtimes for each example is given in Figure \ref{fig:runtime}.

\begin{table}[h]
\begin{center}
\begin{tabular}{|c|c|c|c|c|}
\hline 
$N_{0}$ & $\delta$ & $\mu$ & $\eta$ & $\Theta$ \tabularnewline
\hline
$1000$ & $1.0$ & $1.0$ & $0.9$ & $0.5$ \tabularnewline
\hline 
\end{tabular}
\end{center}
\caption{Parameters values used for all examples.}
\label{tab:parameters}
\end{table}

\begin{figure}[h]
\begin{centering}
\includegraphics[width=0.8 \columnwidth]{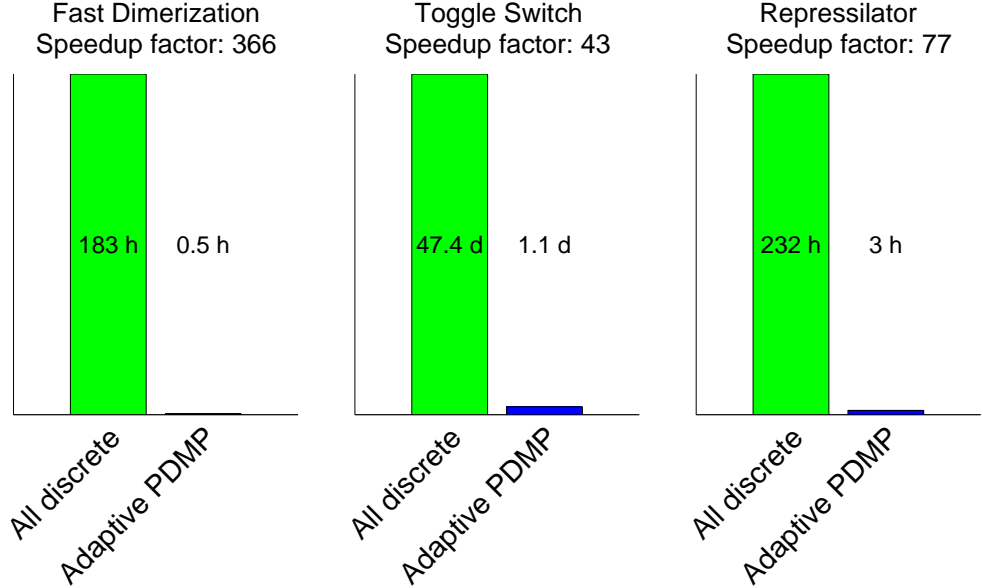} \\
\end{centering}
\caption{Runtime comparison of fully discrete dynamics and our Adaptive PDMP method.
Shown is the comparison of total CPU time required for 100'000 sample paths.}
\label{fig:runtime}
\end{figure}

\subsection{Fast Dimerization}
\label{subsec:fast_dimerization}

The fast dimerization SRN \cite{Cao2005} consists of $n_{S}=3$
species $s_{0}$, $s_{1}$ and $s_{2}$ and $n_{R}=4$ reactions and is depicted in Figure \ref{fig:fastdimerization_schematic}.
The reactions are listed in Table \ref{tab:fastdimerization}.

\begin{table}[h]
\begin{center}
\begin{tabular}{|c|c||c|c|}
\hline 
Reaction & $\kappa'$ & Reaction & $\kappa'$\tabularnewline
\hline 
\hline 
$s_{0} + s_{0} \rightarrow s_{1}$ & $1.0$ & $s_{1} \rightarrow s_{0} + s_{0}$ & $200.0$ \tabularnewline
\hline 
$s_{0} \rightarrow \emptyset$ & $0.02$ & $s_{1} \rightarrow s_{2}$ & $0.004$ \tabularnewline
\hline 
\end{tabular}
\end{center}
\caption{Reactions of the Fast Dimerization network}
\label{tab:fastdimerization}
\end{table}

\begin{figure}[h]
\begin{centering}
\includegraphics[width=0.3 \columnwidth]{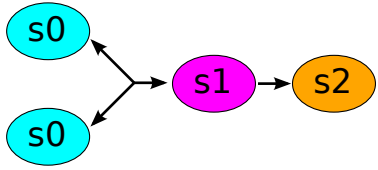} \\
\end{centering}
\caption{Cartoon of the Fast Dimerization network}
\label{fig:fastdimerization_schematic}
\end{figure}

The simulation is run from $t_{0}=0$ until $t_{f}=400$
and the initial state is set to $x\left(t_{0}\right)=x_{0} = \left(540, 730, 0\right)$.
Considering the difference in the orders of magnitude of the rate constants
for reaction $s_0 + s_0 \rightarrow s_1$ and $s_1 \rightarrow s_0 + s_0$ compared
to the rate constants for reaction $s_0 \rightarrow \emptyset$ and $s_1 \rightarrow s_2$,
it seems plausible that the dynamics of a subnetwork consisting of
the species $s_{0}$ and $s_{1}$ can be averaged.
Indeed our adaptive
scheme identifies the subnetwork consisting of species $s_{0}$ and
$s_{1}$ as suitable for averaging, and using the strategy for zero-deficiency networks
the quasi-stationary assumption can be applied.
This can be clearly seen in Figure \ref{fig:fastdimerization_single_mean}, where
the copy-numbers of species $s_{0}$ and $s_{1}$ show strong fluctuations for the SSA
simulations and only weak fluctuations for the Adaptive PDMP simulations (the weak fluctuations
are due to the reaction $s_{1} \rightarrow s_{2}$ which is still simulated with discrete dynamics).
The averaging provides a huge performance
gain without losing any essential information about the dynamics of the species $s_{2}$.

The simulation of $100'000$ sample paths with SSA took a total CPU
time of $\approx 183\text{ hours}$.
The simulation of $100'000$ sample paths with our Adaptive PDMP scheme
using zero-deficiency averaging took a total CPU time of $\approx 0.5\text{ hours}$.
Figure \ref{fig:fastdimerization_single_mean} shows the results of these
simulations.
The distribution given by the CME (estimated with SSA)
closely matches the distribution estimated with our Adaptive PDMP scheme
at time $t=200$ (see Figure \ref{fig:fastdimerization_dist}).
The Kolmogorov-Smirnov distance between the two distributions is $d_{KS}=0.01$.

\begin{figure}[h]
\begin{centering}
\includegraphics[width=\columnwidth]{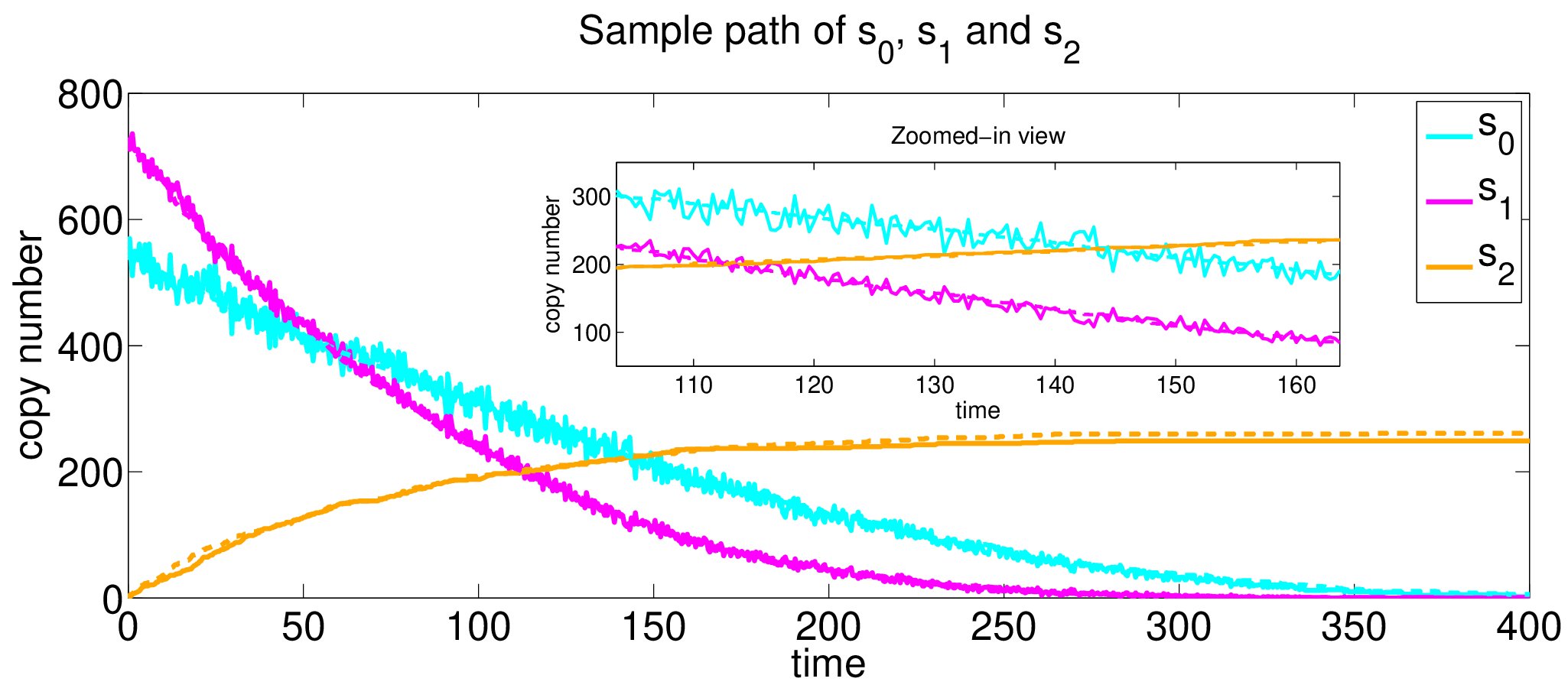} \\
\includegraphics[width=\columnwidth]{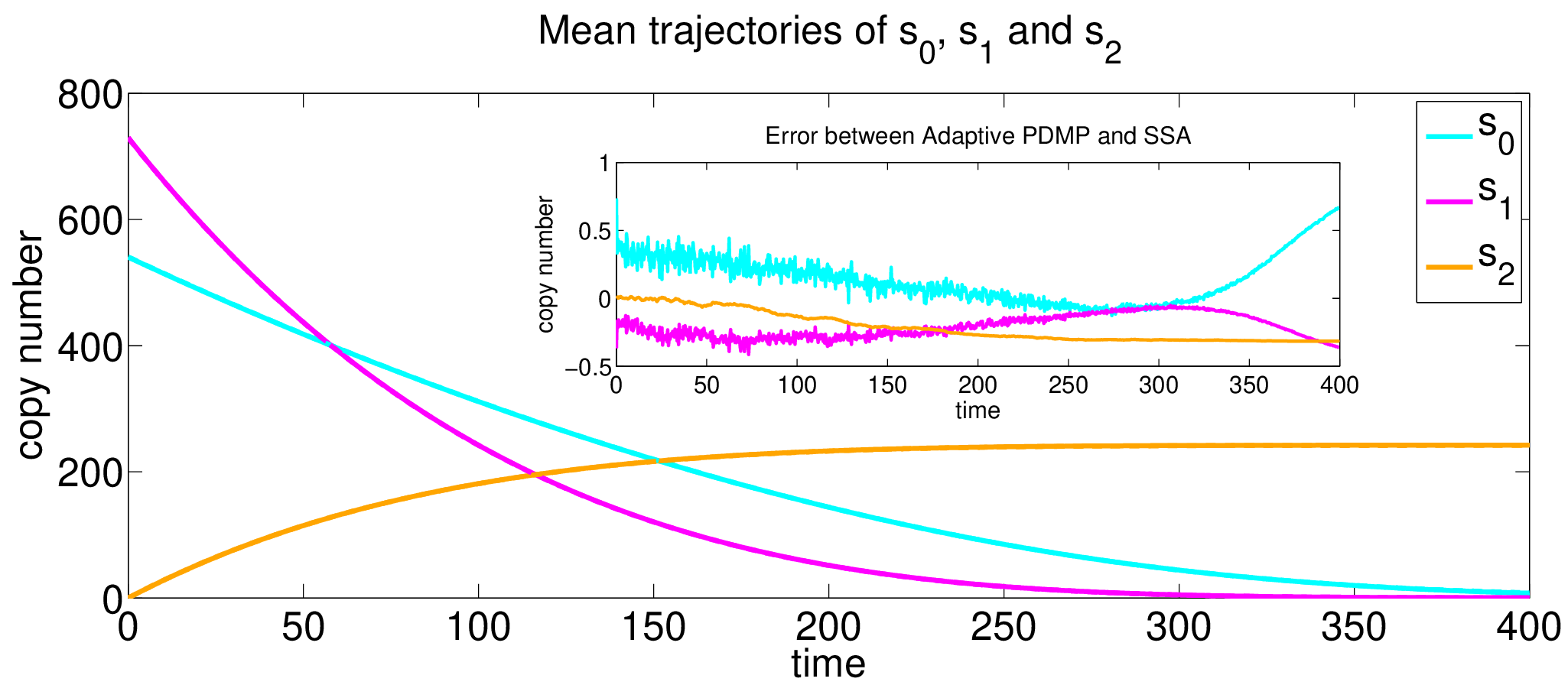} \\
\end{centering}
\caption{Simulation results of the Fast Dimerization SRN}
{\small The upper plot shows a sample path of species $s_0$, $s_1$ and $s_2$ with
SSA (solid lines) and with our Adaptive PDMP scheme (dashed lines).}
{\small The lower plot shows the mean copy-number of species $s_0$, $s_1$ and $s_2$
of $100'000$ simulations with SSA and our Adaptive PDMP scheme.}
\label{fig:fastdimerization_single_mean}
\end{figure}

\begin{figure}[h!]
\begin{centering}
\includegraphics[width=\columnwidth]{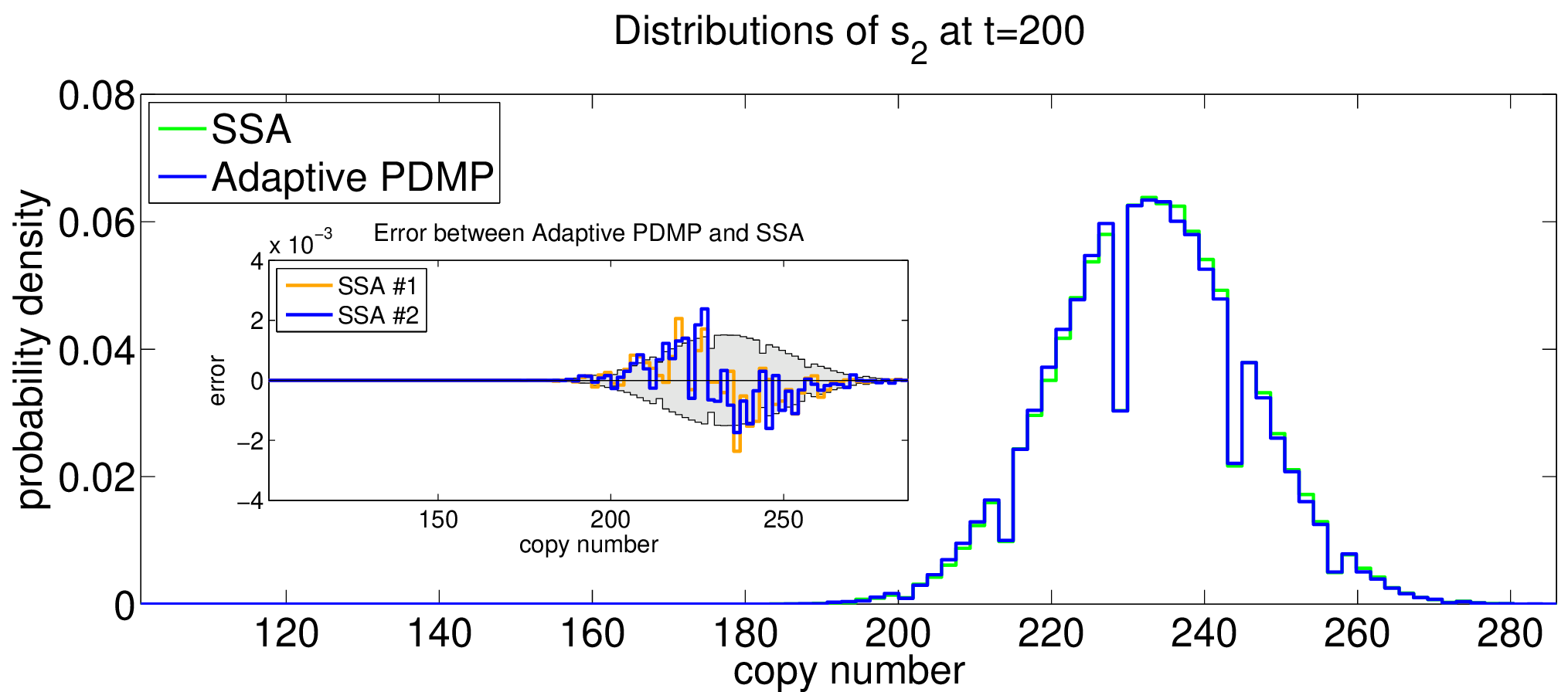} \\
\includegraphics[width=\columnwidth]{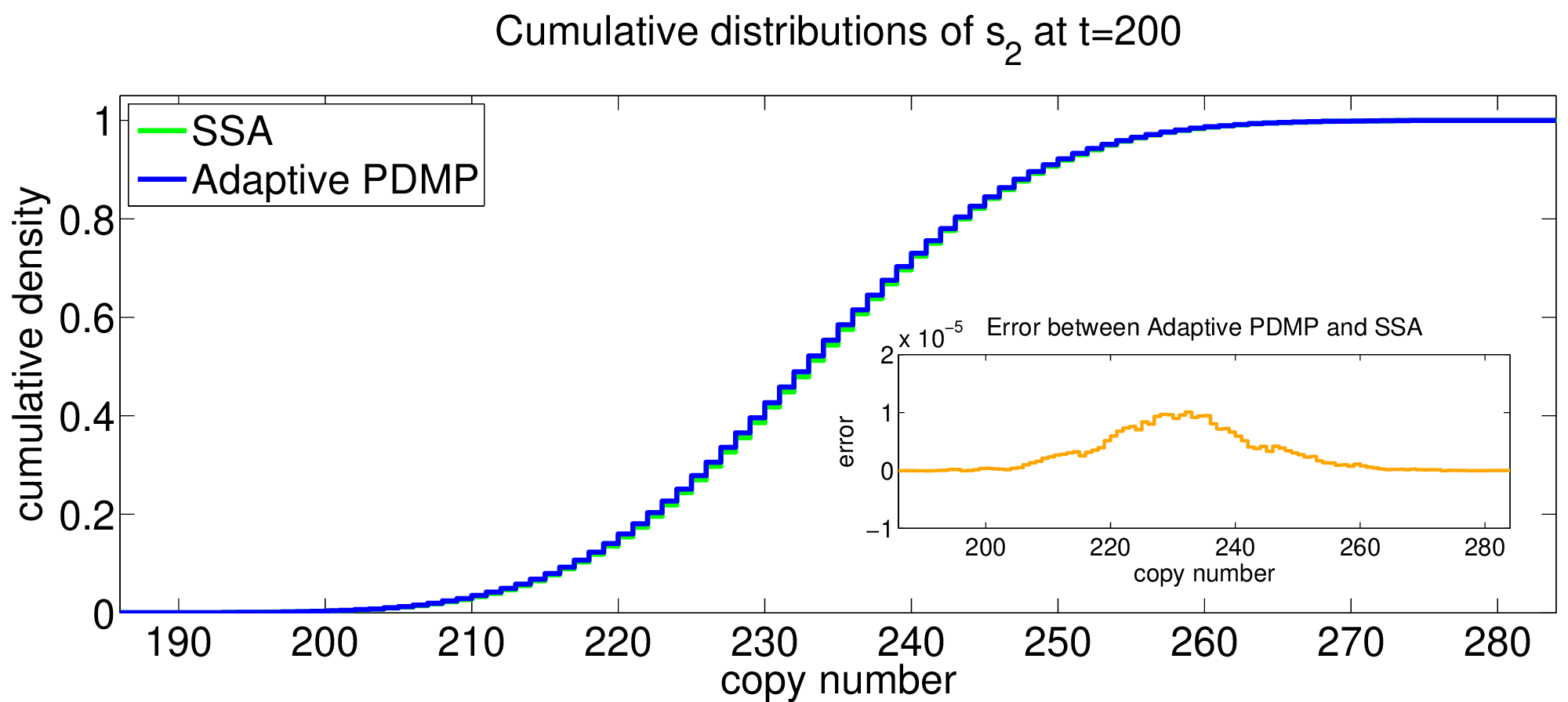} \\
\end{centering}
\caption{Probability distribution of the Fast Dimerization SRN}
{\small The upper plot shows the copy-number distribution (100 equally spaced bins)
of species $s_{2}$ at time $t=200$ with SSA (green) and with our Adaptive PDMP scheme (blue).
The inset shows the error compared to two independent runs with SSA,
each with $100'000$ samples. The grey shaded area
marks the 95\% confidence interval of the SSA run \#1.}{\small \par}
{\small The lower plot shows the cumulative copy-number distributions of species $s_{2}$ at time $t=200$
with SSA (green) and with our Adaptive PDMP scheme (blue).}
\label{fig:fastdimerization_dist}
\end{figure}

\subsubsection{Comparison to an existing adaptive PDMP method}

We compare our Adaptive PDMP scheme with an existing adaptive PDMP method proposed by Alfonsi et al.\ \cite{alfonsi2005}.
We implemented the method similarly to our own implementation and added the necessary details.
We compare both methods with SSA for the Fast Dimerization SRN without averaging.

Figure \ref{fig:fastdimerization_dist_comparison} shows the results of $10'000$ sample paths for each method. One can clearly see a good match of our method with
SSA. The results from the method by Alfonsi et al.\ are inaccurate. In this case, we suspect that this is due to the partitioning criteria that only considers the copy-number
of reactants. Because of this the reaction $s_{1} \rightarrow s_{2}$ is considered as being continuous when the copy-number of $s_{1}$ is big enough. Thus
the stochastic variation in species $s_{2}$ is strongly reduced in comparison to the exact SSA simulations.

\begin{figure}[h!]
\begin{centering}
\includegraphics[width=\columnwidth]{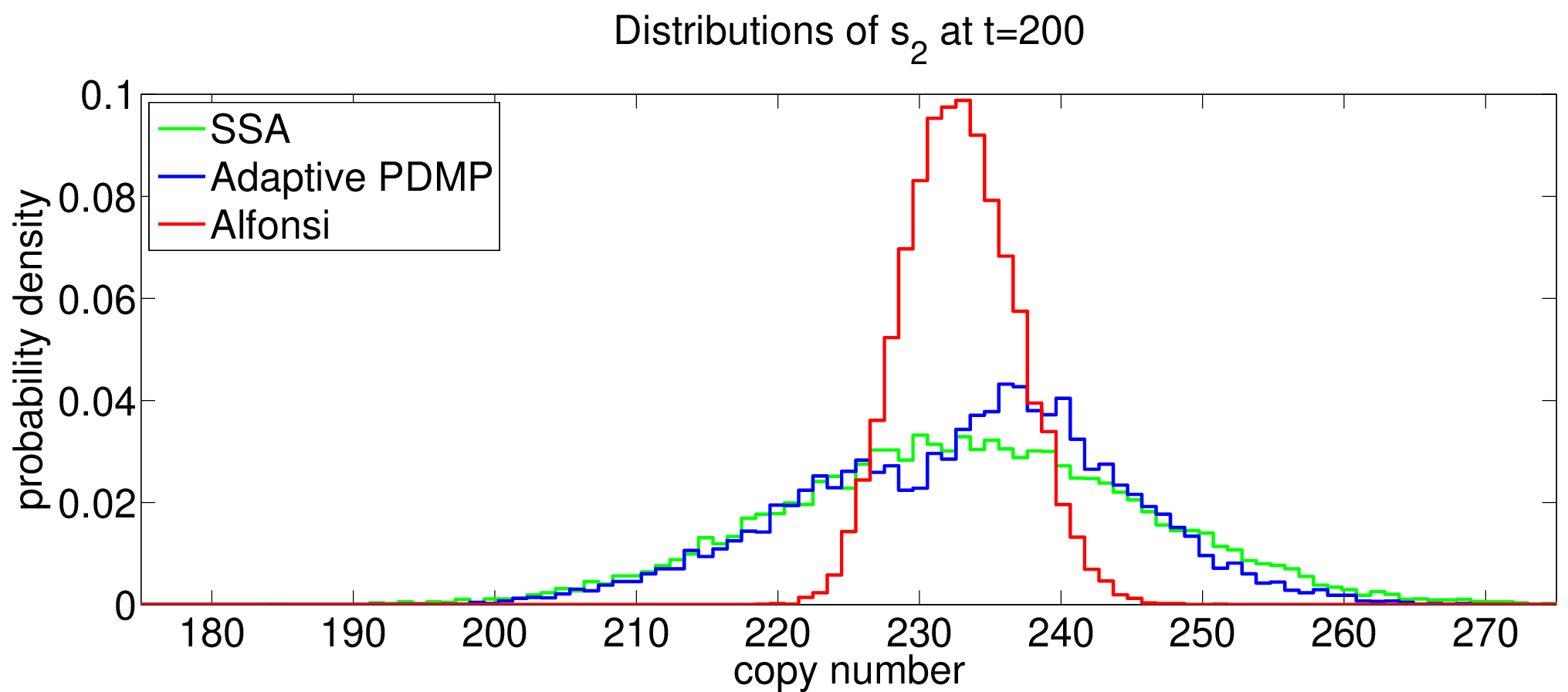} \\
\includegraphics[width=\columnwidth]{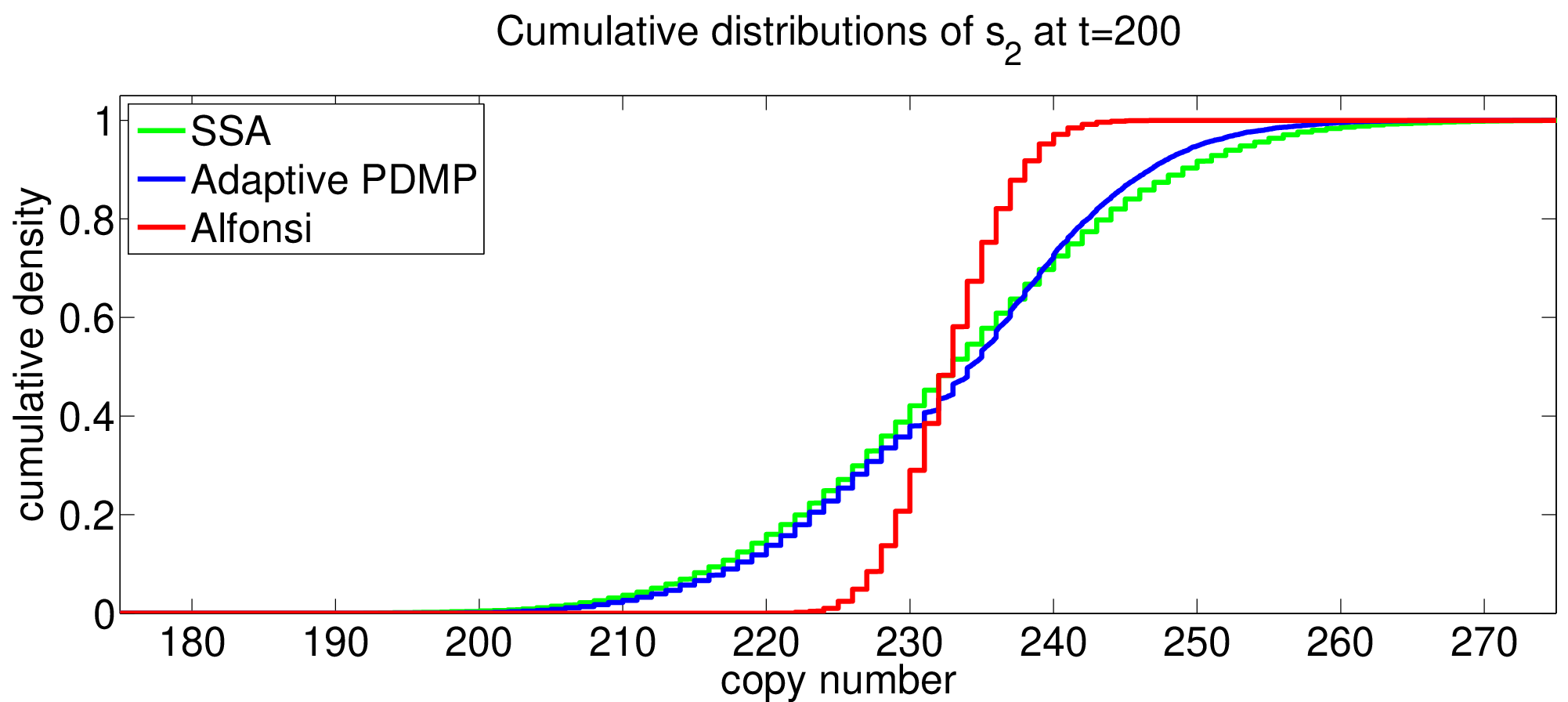} \\
\end{centering}
\caption{Comparison of the probability distribution of the Fast Dimerization SRN}
{\small
The upper plot shows the copy-number distribution (100 equally spaced bins)
of species $s_{2}$ at time $t=200$ with SSA (green), with our Adaptive PDMP (blue) and with the Alfonsi scheme (red).
}{\small \par}
{\small The lower plot shows the cumulative copy-number distributions of species $s_{2}$ at time $t=200$
with SSA (green), with our Adaptive PDMP (blue) and with the Alfonsi scheme (red).
}
\label{fig:fastdimerization_dist_comparison}
\end{figure}

\subsection{Toggle Switch}

We consider the Toggle-Switch network, which is similar to
the synthetic Toggle-Switch by Gardner et al.\cite{Gardner2000},
but implemented with mass-action kinetics.
It consists of $n_{S}=6$ species $m_{A}$, $s_{A}$, $m_{B}$, $s_{B}$,
$p_{A}$, $p_{B}$ and $n_{R}=16$ reactions and is depicted in Figure \ref{fig:toggleswitch_schematic}.
The reactions are listed in Table \ref{tab:toggleswitch}.

\begin{table}[h]
\begin{center}
\begin{tabular}{|c|c||c|c||cc}
\hline
Reaction & $\kappa'$ & Reaction & $\kappa'$ & \multicolumn{1}{|c|}{Reaction} & \multicolumn{1}{|c|}{$\kappa'$} \tabularnewline
\hline
\hline
$\emptyset\ \rightarrow\ m_{A}$ & $1.0$ & $m_{A}\rightarrow s_{A}$ & $5.0$ & \multicolumn{1}{|c|}{$s_{A}\rightarrow\emptyset$} & \multicolumn{1}{|c|}{$0.01$} \tabularnewline
\hline
$\emptyset\rightarrow m_{B}$ & $1.0$ & $m_{B}\rightarrow s_{B}$ & $5.0$ & \multicolumn{1}{|c|}{$s_{B}\rightarrow\emptyset$} & \multicolumn{1}{|c|}{$0.01$} \tabularnewline
\hline
$m_{A}\rightarrow\emptyset$ & $0.1$ & $s_{A}+m_{B}\rightarrow s_{A}$ & $20.0$ & \multicolumn{1}{|c|}{$s_{A}\rightarrow s_{A}+p_{A}$} & \multicolumn{1}{|c|}{$10.0$} \tabularnewline
\hline
$m_{B}\rightarrow\emptyset$ & $0.1$ & $s_{B}+m_{A}\rightarrow s_{B}$ & $20.0$ & \multicolumn{1}{|c|}{$s_{B}\rightarrow s_{B}+p_{B}$} & \multicolumn{1}{|c|}{$10.0$} \tabularnewline
\hline
$p_{A}\rightarrow\emptyset$ & $0.1$ & $p_{B}\rightarrow\emptyset$ & $0.1$ & &\tabularnewline
\cline{1-4}
\end{tabular}
\end{center}
\caption{Reactions of the Toggle Switch network}
\label{tab:toggleswitch}
\end{table}

\begin{figure}[h!]
\begin{centering}
\includegraphics[width=0.3 \columnwidth]{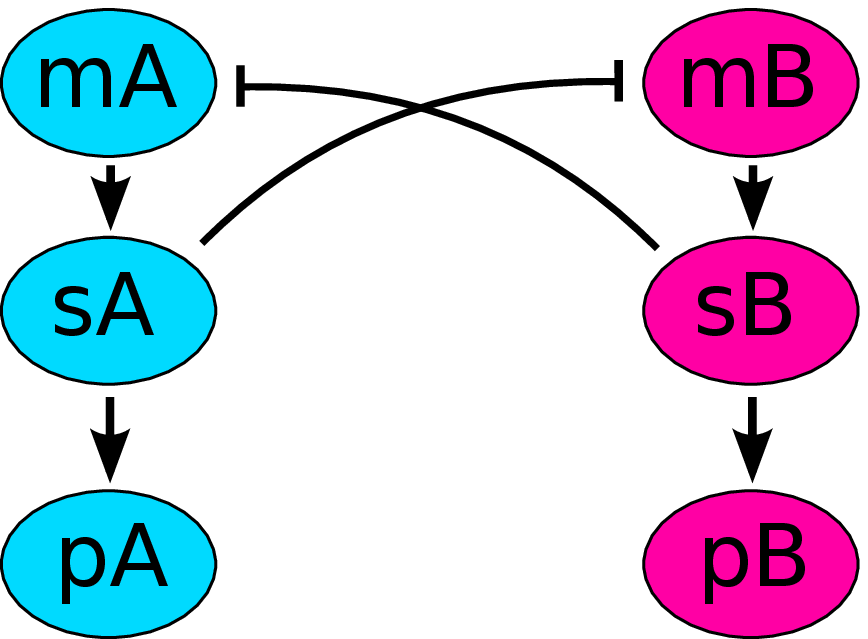} \\
\end{centering}
\caption{Cartoon of the Toggleswitch network}
\label{fig:toggleswitch_schematic}
\end{figure}

We estimate the solution given by the CME at time $t_{f} = 10^5$
where the initial state is set to $x_{0}=\left(0,\ldots,0\right)$.
\trred{As our scheme could not find any suitable subnetworks for this example,
we turned off the averaging procedure to reduce the computational overhead.}
There are two high copy-number proteins, $p_{A}$ and $p_{B}$, which
are controlled by two low copy-number mRNAs, $s_{A}$ and $s_{B}$.
These mRNAs $s_{A}$ and $s_{B}$ are the processed variants of the precursor mRNAs $m_{A}$ and $m_{B}$.
In addition to inducing translation the mRNAs, $s_{A}$ and $s_{B}$ also induce the degradation of $m_{A}$ and $m_{B}$ respectively.
The network shows a bistable behaviour and can stochastically
switch between a high copy-number $p_{A}$, low copy-number $p_{B}$ state
and a low copy-number $p_{A}$, high copy-number $p_{B}$ state. While the
SSA will spend a major part of the simulation time for simulating
the translation and degradation of the $p_{A}$ and $p_{B}$ proteins, our adaptive
PDMP scheme performs those reactions with continuous dynamics when the proteins
have a high copy-number, yielding a considerable performance gain.

The simulation of $100'000$ sample paths with SSA took a total CPU
time of $\approx47.4\text{ days}$.
The simulation of $100'000$ sample paths with our Adaptive PDMP scheme
using a step size of $dt = 5.0$ took a total CPU time of $\approx1.1\text{ days}$.
Figure \ref{fig:toggleswitch} shows the results of these simulations.
The distribution given by the CME (estimated with SSA)
closely matches the distribution estimated with our Adaptive PDMP scheme
at time $t=25'000$.
The Kolmogorov-Smirnov distance between the two distributions is $d_{KS}=0.01$.

\begin{figure}[h!]
\begin{centering}
\includegraphics[width=\columnwidth]{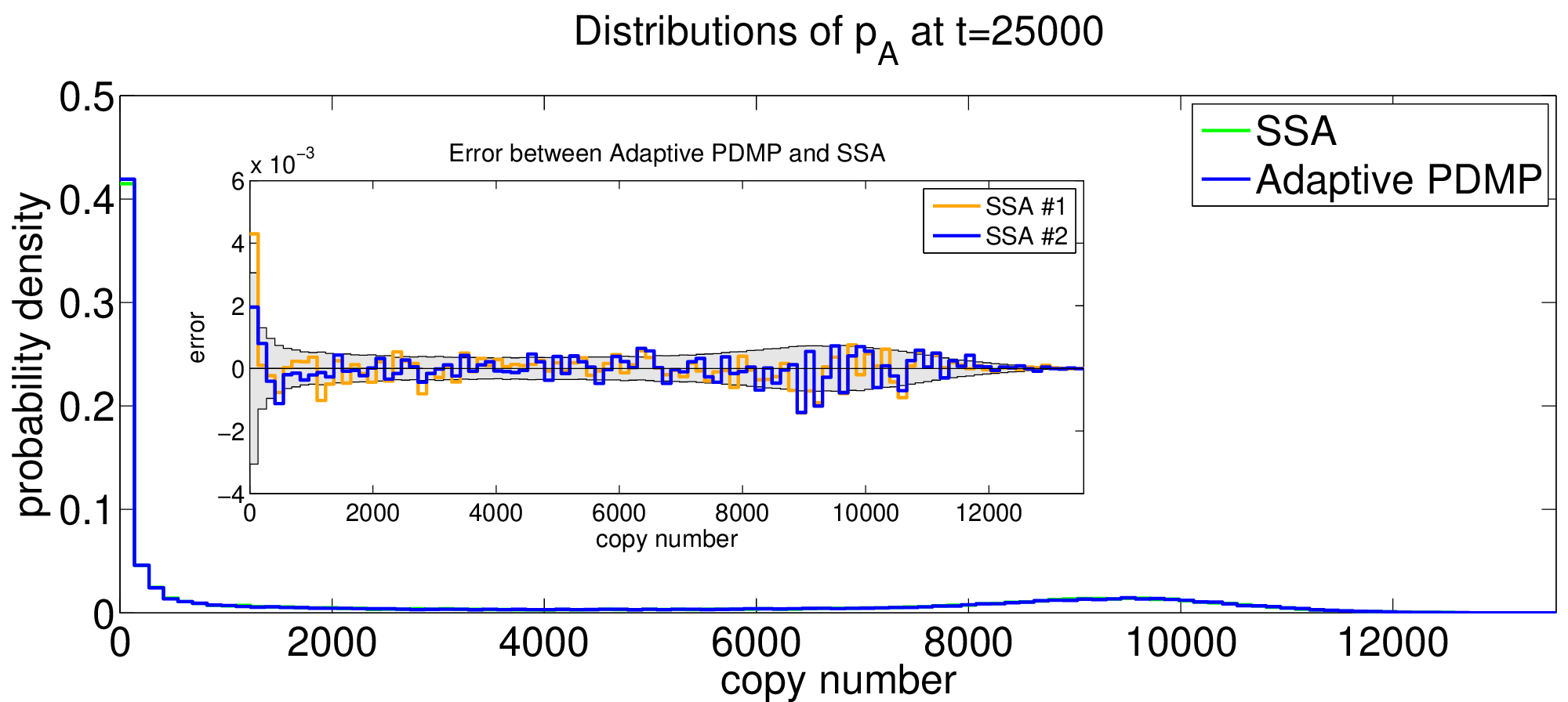} \\
\includegraphics[width=\columnwidth]{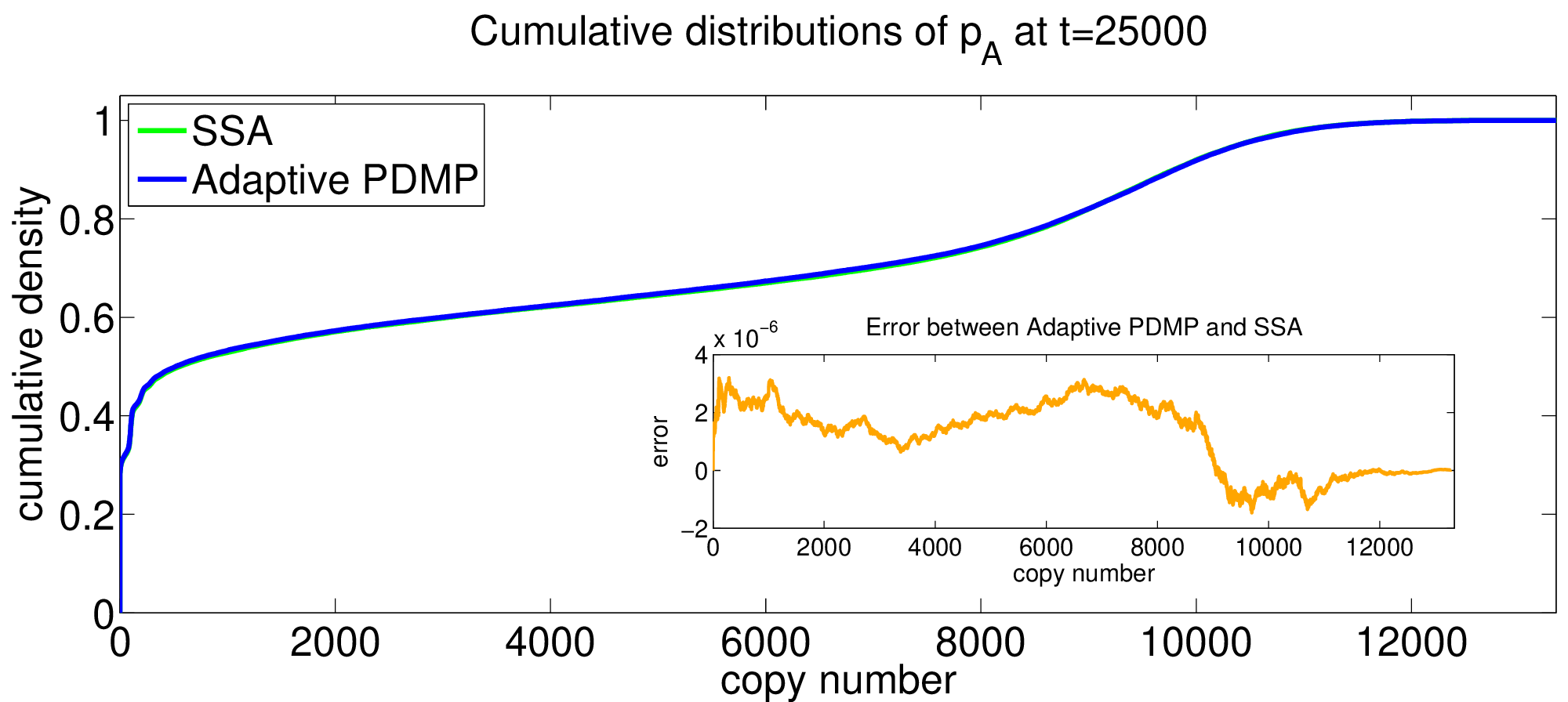}
\end{centering}
\caption{Probability distribution of the Toggleswitch SRN}
{\small The upper plot shows the copy-number distribution (100 equally spaced bins)
of species $p_{A}$ at time $t=25'000$ with SSA (green) and with our Adaptive PDMP scheme (blue).
The inset shows the error compared to two independent runs with SSA,
each with 100'000 samples. The grey shaded area
marks the 95\% confidence interval of the SSA run \#1.}
{\small The lower plot shows the cumulative copy-number distributions of species $p_{A}$ at time $t=25'000$
with SSA (green) and with our Adaptive PDMP scheme (blue).}
\label{fig:toggleswitch}
\end{figure}

\subsection{Repressilator}

We consider the Repressilator network, which is similar to
the synthetic Repressilator by Elowitz et al.\cite{elowitz2000synthetic},
but implemented with mass-action kinetics.
It consists of $n_{S}=6$ species $m_{A}$, $p_{A}$, $m_{B}$, $p_{B}$,
$m_{C}$, $p_{C}$ and $n_{R}=15$ reactions and is depicted in Figure \ref{fig:repressilator_schematic}.
The reactions are listed in Table \ref{tab:repressilator}.

\medskip{}

\begin{table}[h]
\begin{center}
\begin{tabular}{|c|c||c|c||c|c|}
\hline 
Reaction & $\kappa'$ & Reaction & $\kappa'$ & Reaction & $\kappa'$\tabularnewline
\hline 
\hline 
$\emptyset\rightarrow m_{A}$ & $0.1$ & $\emptyset\rightarrow m_{B}$ & $0.1$ & $\emptyset\rightarrow m_{C}$ & $0.1$\tabularnewline
\hline 
$m_{A}\rightarrow m_{A}+p_{A}$ & $50.0$ & $m_{B}\rightarrow m_{B}+p_{B}$ & $50.0$ & $m_{C}\rightarrow m_{C}+p_{C}$ & $50.0$\tabularnewline
\hline 
$m_{A}\rightarrow\emptyset$ & $0.01$ & $m_{B}\rightarrow\emptyset$ & $0.01$ & $m_{C}\rightarrow\emptyset$ & $0.01$\tabularnewline
\hline 
$m_{A}+p_{B}\rightarrow p_{B}$ & $50.0$ & $m_{B}+p_{C}\rightarrow p_{C}$ & $50.0$ & $m_{C}+p_{A}\rightarrow p_{A}$ & $50.0$\tabularnewline
\hline 
$p_{A}\rightarrow\emptyset$ & $0.01$ & $p_{B}\rightarrow\emptyset$ & $0.01$ & $p_{C}\rightarrow\emptyset$ & $0.01$\tabularnewline
\hline 
\end{tabular}
\end{center}
\caption{Reactions of the Repressilator network}
\label{tab:repressilator}
\end{table}

\begin{figure}[h!]
\begin{centering}
\includegraphics[width=0.3 \columnwidth]{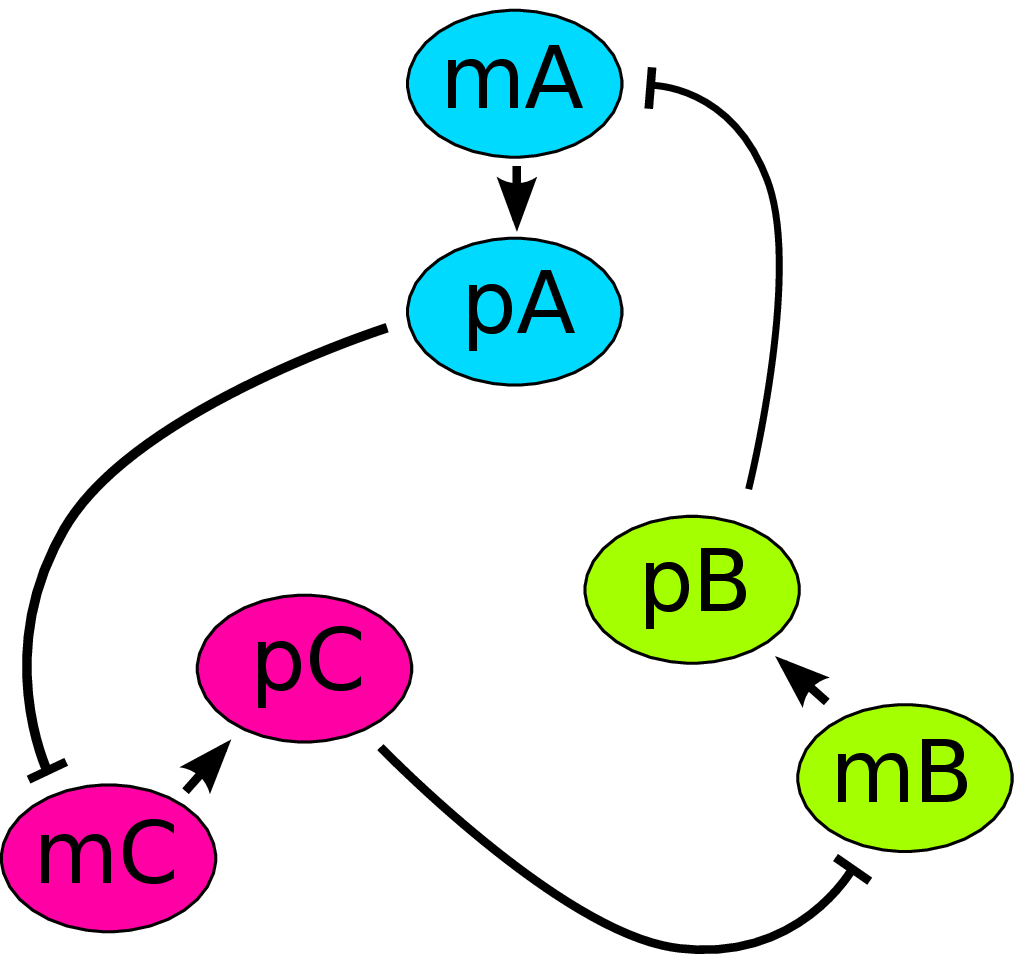} \\
\end{centering}
\caption{Cartoon of the Repressilator network}
\label{fig:repressilator_schematic}
\end{figure}

We estimate the solution given by the CME at time $t_{f} = 5 \times 10^4$
where the initial state is set to $x_{0}=\left(10, 500, 0, 0, 0, 0\right)$.
\trred{Also in this example, there were no suitable subnetworks for the averaging procedure,
so we reduced the computational overhead by turning it off.}

The Repressilator consists of three mRNAs $m_{A}$, $m_{B}$ and $m_{C}$
and the three corresponding proteins $p_{A}$, $p_{B}$ and $p_{C}$.
Each protein catalyses the degradation of another mRNA in a circular
manner, that is, $p_{A}$ degrades $m_{C}$, $p_{B}$ degrades $m_{A}$
and $p_{C}$ degrades $m_{B}$. In the stochastic setting, this network exhibits spontaneous
oscillatory behaviour (see Figure \ref{fig:repressilator_adaptation}).

\begin{figure}[h!]
\begin{centering}
\includegraphics[width=\columnwidth]{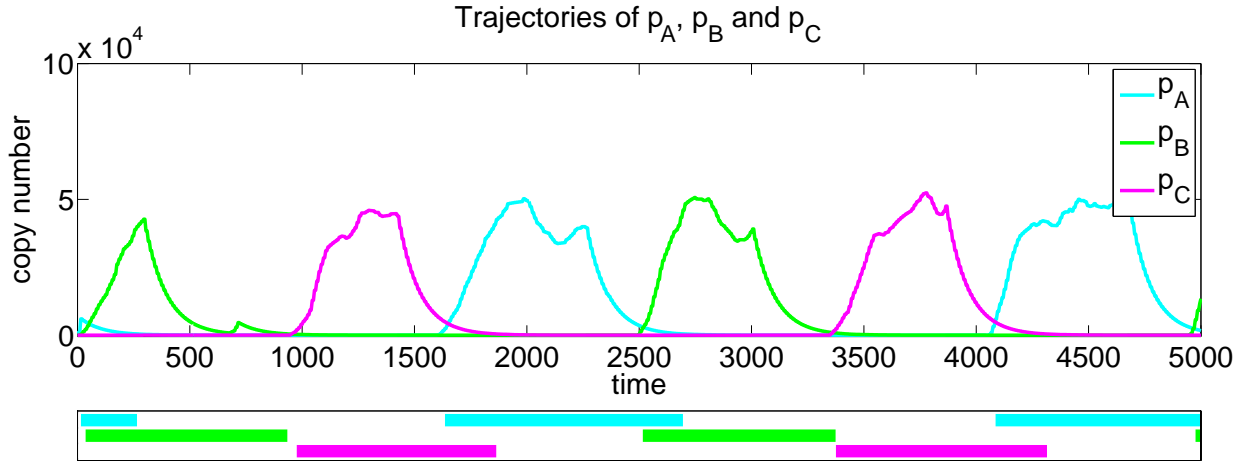} \\
\end{centering}
\caption{Single trajectory showing the adaptation of our method}
{\small Shown is a single sample-path of the copy-number for species $p_{A}$, $p_{B}$ and $p_{C}$ simulated with
our Adaptive PDMP scheme. The colored areas below indicate when the corresponding species $p_{X}$ and the
reactions $m_{X} \rightarrow m_{X} + p_{X}$ and $p_{X} \rightarrow \emptyset$ are considered to be continuous.}{\small \par}
\label{fig:repressilator_adaptation}
\end{figure}

\trred{
A fixed PDMP scheme that simulates reactions affecting only the protein species continuously
(i.e. $p_{X} \rightarrow \emptyset$ and $m_{X} \rightarrow m_{X} + p_{X}$) and the rest discretely,
does not exhibit the oscillatory behaviour that can be seen by simulating the SRN with SSA.
}
This can be seen in Figure \ref{fig:repressilator_mean} where the first moment of the copy-number of
species $p_{A}$ is shown for different simulation schemes. The result of a fixed PDMP scheme that
simulates the translation- and degradation-reaction of species $p_{A}$ with continuous
dynamics is shown in red and shows a quick drop of the copy-number to $0$.
This is due to the discrete nature of the species in the stochastic
description. In the fixed PDMP scheme the protein $p_{A}$ with
a very small concentration can still degrade the mRNA $m_{C}$,
whereas in the stochastic setting the copy-number of protein $p_{A}$ can
drop to $0$ and will not contribute to
the degradation of $m_{C}$ anymore, thus enabling $m_{C}$ to rise.
On the other hand, when the copy-numbers of the proteins are high,
the stochasticity of the translation- and the degradation-reaction is negligible.

Our Adaptive PDMP scheme can approximate the protein dynamics continuously
when they have high copy-numbers and switch to discrete stochastic
simulations when they have low copy-numbers.
This is depicted in Figure \ref{fig:repressilator_adaptation} where the colored bars
indicate when the corresponding degradation- and translation-reactions are approximated
with continuous dynamics. One can easily see that the dynamics for the
degradation- and translation-reaction of $p_{B}$ is approximated with continuous
dynamics when the copy-number of $p_{B}$ is high.
In Figure \ref{fig:repressilator_mean} of simulations with SSA (green) and with our adaptive method (blue) show the same oscillatory behaviour.
\trred{
Due to stochasticity the average amplitude decreases with time as the individual trajectories
lose their synchronization, unlike the deterministic setting. Note that the initial state is the same for each trajectory.
}
Hence this example illustrates, how an adaptive partitioning scheme can be very useful.

The simulation of $100'000$ sample paths with SSA took a total CPU
time of $\approx232\text{ hours}$.
The simulation of $100'000$ sample paths with our Adaptive PDMP scheme
using a step size of $dt = 1.0$ took a total CPU time of $\approx3\text{ hours}$.
Figure \ref{fig:repressilator} shows the distribution of species $p_{A}$ at time $t=4'750$ for simulations with SSA (green)
and with our adaptive scheme (blue). In addition the insets show the deviation
from two independent runs with SSA and also mark the 95\% confidence intervals (grey area).
The distribution given by the CME (estimated with SSA)
closely matches the distribution estimated with our Adaptive PDMP scheme
at time $t=4'750$.
The Kolmogorov-Smirnov distance between the two distributions is $d_{KS}=0.003$.

\begin{figure}[h!]
\begin{centering}
\includegraphics[width=\columnwidth]{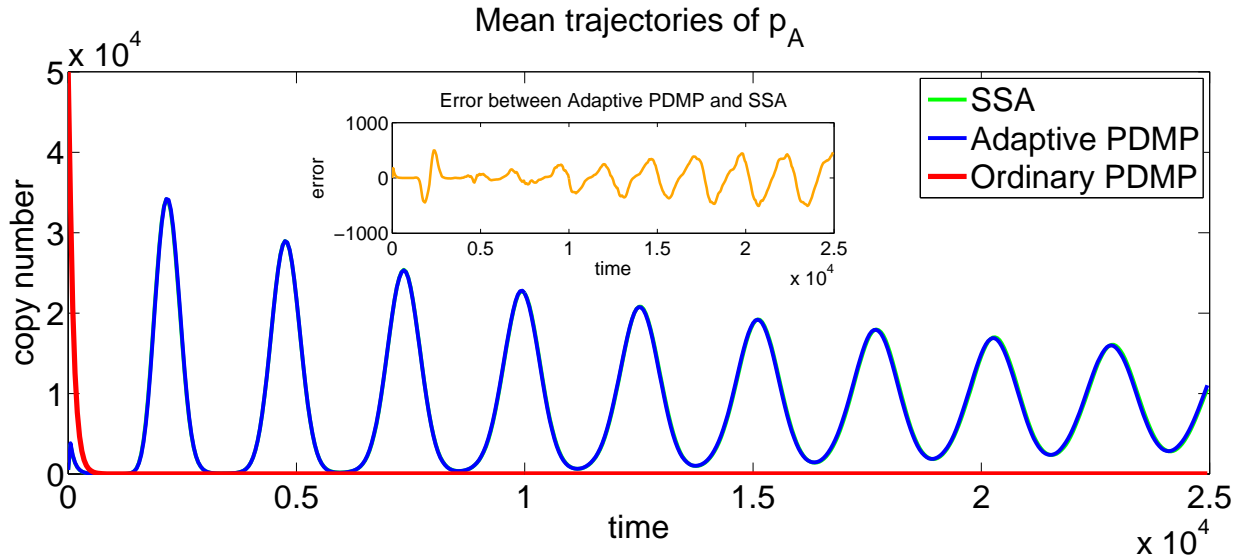} \\
\end{centering}
\caption{Mean trajectories of species $p_{A}$ for the Repressilator SRN}
{\small Shown is the mean of species $p_{A}$ of
$100'000$ simulations with SSA (green), with the fixed
PDMP scheme (red) and with our Adaptive PDMP scheme (blue).}
\label{fig:repressilator_mean}
\end{figure}

\begin{figure}[h!]
\begin{centering}
\includegraphics[width=\columnwidth]{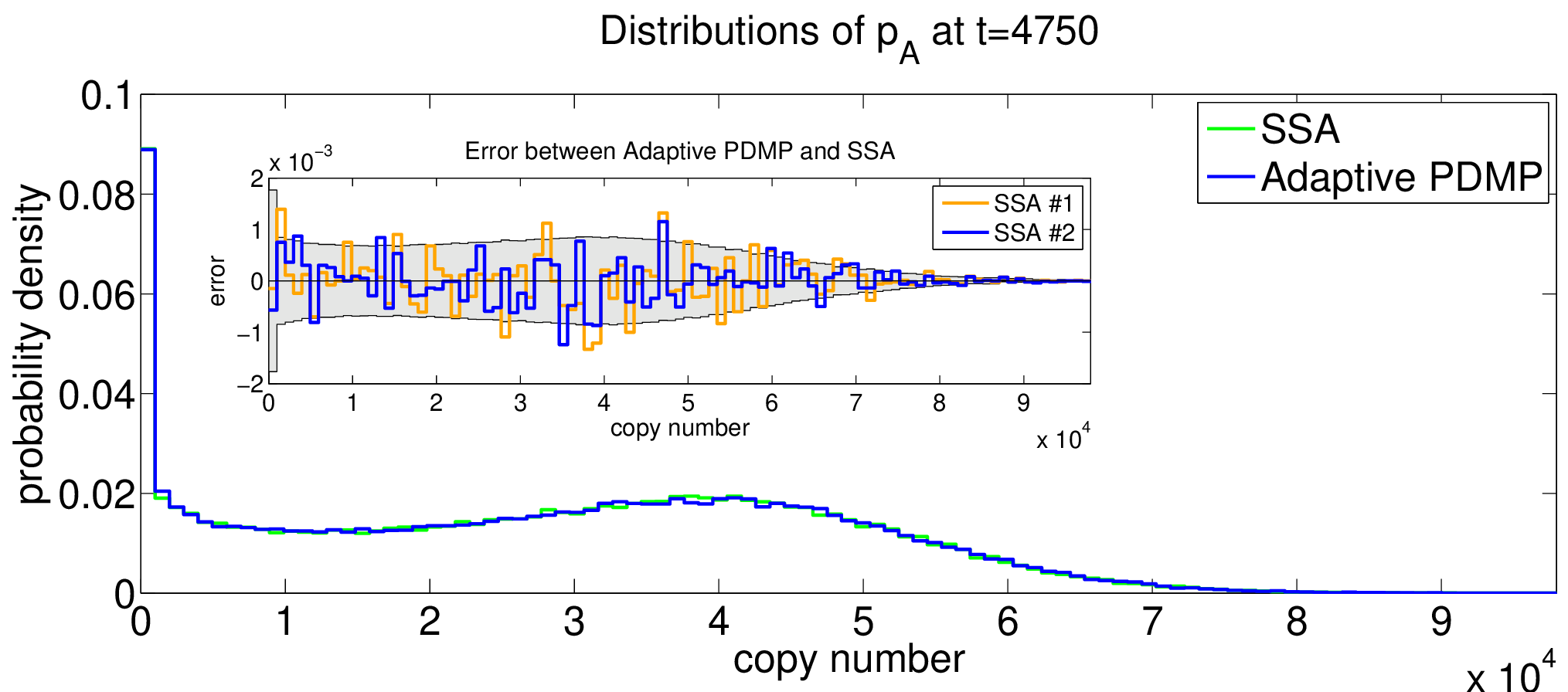} \\
\includegraphics[width=\columnwidth]{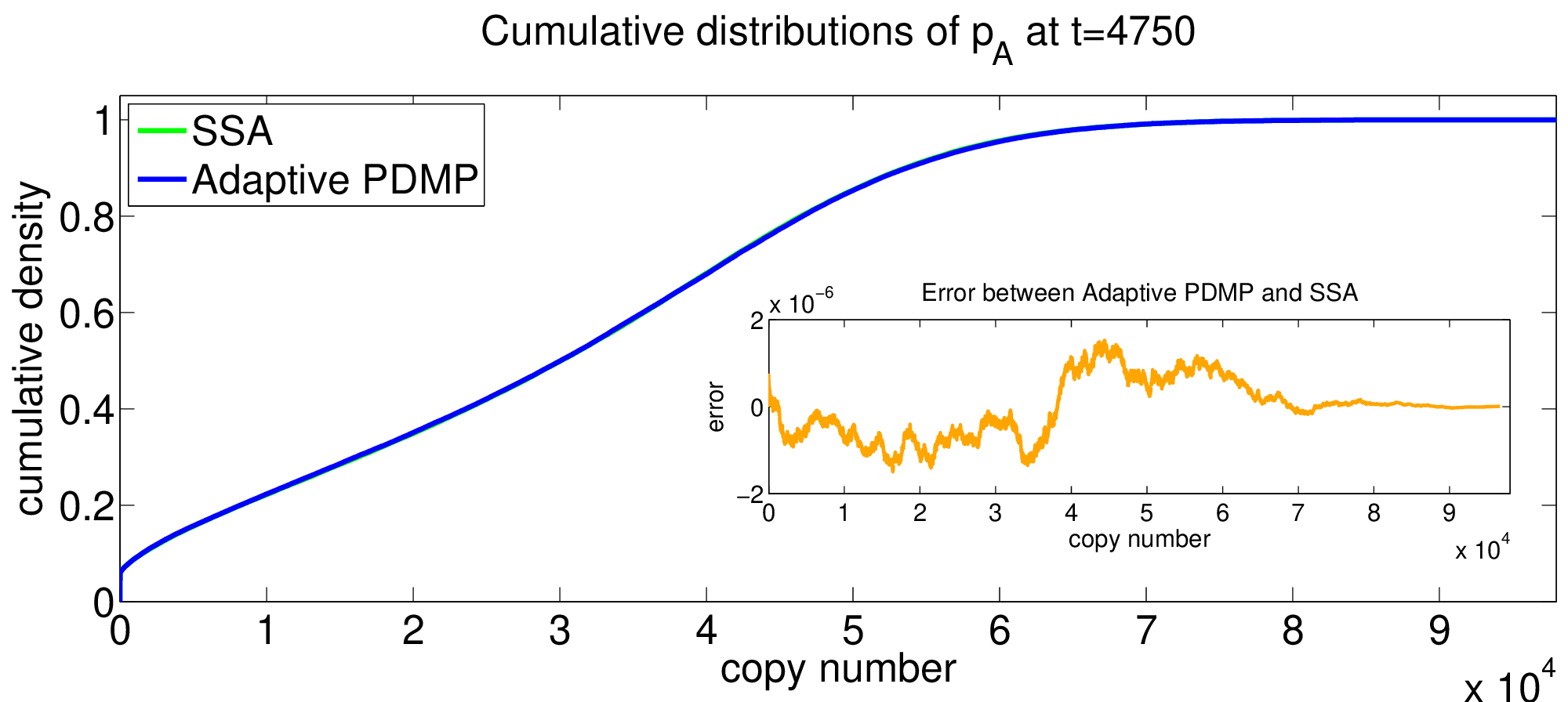} \\
\end{centering}
\caption{Simulation results of the Repressilator SRN}
{\small The upper plot shows the copy-number distribution (100 equally spaced bins)
of species $p_{A}$ at time
$t=4'750$ with SSA (green) and with our Adaptive PDMP scheme (blue).
The inset shows the error compared to two independent runs with SSA,
each with 100'000 samples. The grey shaded area
marks the 95\% confidence interval of the SSA run \#1.}{\small \par}
{\small The lower plot shows the cumulative copy-number distributions of species $p_{A}$ at time $t=4'750$
with SSA (green) and with our Adaptive PDMP scheme (blue).}
\label{fig:repressilator}
\end{figure}

\section{Conclusions and future work\label{sec:CONCLUSIONS}}

In this paper we propose a novel adaptive hybrid scheme for approximating the
solution of a Chemical Master Equation for Stochastic Reaction Networks spanning a wide range of reaction timescales,
and having large variation in species copy-number scales.
The method is based on a rigorous mathematical framework by Kang et al.\cite{kang2013separation}, that ensures the validity of the approximations.

To achieve considerable speed-ups over SSA our method introduces two sources of error:
we treat high copy-number species as continuous
and we assume stationarity for the dynamics of fast subnetworks.
However with the help of parameters we can tune the amount of error one is willing to tolerate in a qualitative manner.
For example, setting $N_0 \approx \infty$ will reduce our scheme
to SSA and setting $\Theta \approx \infty$ would prevent the application
of the quasi-stationary assumption.

We compare the results of our adaptive method with SSA and a fixed PDMP scheme using
examples from Systems Biology.
Our results show that the adaptive method can provide considerable performance enhancements in comparison to both these approaches.
Moreover, as our last example demonstrates, the adaptivity can be crucial to achieve accuracy and
reduce the computational cost at the same time.

The adaptive feature of our method makes it particularly well suited for SRNs
where the timescales of the reaction dynamics can change with time.
Such behaviour can be expected in regulatory gene expression networks.
We hope that our method will provide a tool for researchers to analyse more complicated and realistic models in Systems Biology.

Our adaptive scheme can potentially be used in the simulation
of large scale compartment models of Spatial SRNs \cite{elf2004spontaneous}.
For such systems the spatial distribution of species can change with time and
our adaptive method could reduce the computational costs significantly.
For example, in the Min System in E. coli \cite{fange2006noise},
the location of high copy-number species is constantly changing. So an adaptive approach could neglect
the fluctuations at these locations, while conserving the fluctuations
at locations with low copy-numbers.
In the future we aim to develop such adaptive schemes for Spatial SRNs.

Another direction for further research is a combination of our method with a $\tau$-leaping scheme, as we propose in Section \ref{sec:COMBINATION_TAU_LEAP},
and the corresponding study of the mathematical validity of such an approach.

\appendix

\section{Averaging of fast subnetworks}
\label{appendix:averaging}

\trred{We now describe the averaging procedure for fast subnetworks in this section which is
an optional extension to the previously discussed approximation of SRNs.}

To check if a subnetwork has \textit{fast} dynamics in comparison to the surrounding network
(which consists of species that are directly influenced by the subnetwork)
we look at two sets of reactions, the reactions within the subnetwork, and the
reactions connecting the subnetwork to the surrounding network.
We then define a timescale for both these sets of reactions
and describe a formal criterion to check whether the subnetwork is fast.

As an example, Figure \ref{fig:averaging_schematic} shows a network where a transcription factor switches
between an active (tf*) and an inactive (tf) form. In the active form it can bind to a gene and
initiate transcription. Assuming that the switch between the active and inactive form of the transcription factor is fast in comparison to
the binding of the active transcription factor to the gene, the red reaction would define a fast subnetwork. The orange species make up
the fast subnetwork and the blue species make up the surrounding network.
The green reaction connects the subnetwork to the surrounding network.

\begin{figure}[h]
\begin{centering}
\includegraphics[width=0.5 \columnwidth]{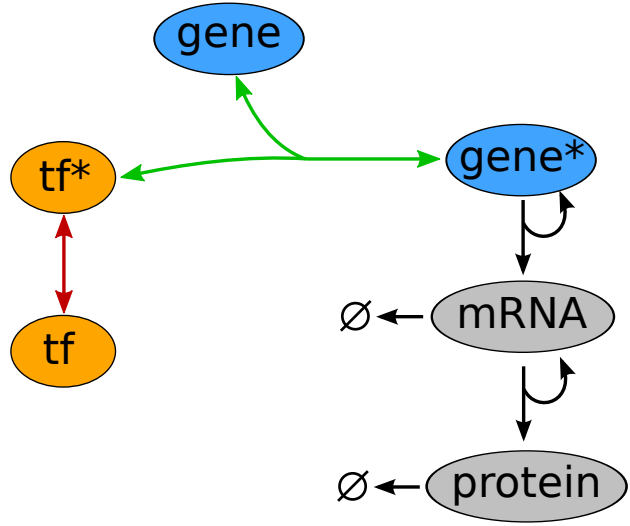} \\
\end{centering}
\caption{Example network depicting a fast subnetwork and the corresponding reactions}
\label{fig:averaging_schematic}
\end{figure}

Now we describe this procedure elaborately.
\trred{Consider a subnetwork with $m_{R}$ reactions $R \subseteq \left\{1,\ldots,n_{R}\right\}$.} We define
the set of species involved in the reactions in $R$ as
\[
Q\left(R\right) = \left\{ i: \xi_{ik} \neq 0 \text{ for some } k \in R \right\} 
\]
and define the set of species that are catalytically involved in the reactions in $R$ as
\trred{
\[
Q_{C}\left(R\right) = \left\{ i \not\in Q\left(R\right): \nu_{ik} = \nu'_{ik} > 0 \text{ for some } k \in R \right\}.
\]
}
Note that the propensities of the reactions in $R$ only depend on the species in $Q\left(R\right)$
and $Q_{C}\left(R\right)$.
Assume that the copy-numbers of the species in $Q_{C}\left(R\right)$ are constant. Then the reactions
in $R$ and the species in $Q\left(R\right)$ naturally define a SRN.
Suppose that the dynamics of this SRN converges to a unique stationary
distribution for any choice of copy-numbers of the species in $Q_{C}\left(R\right)$.
If the reactions in $R$ are sufficiently faster than
the surrounding reactions, given by
\begin{align*}
R_{S} \left(R\right) = & \left\{ k\not\in R:\nu_{ik} > 0 \text{ or } \nu'_{ik} > 0 \text{ for some }i\in Q\left(R\right)\right\} \\
\cup & \left\{ k\not\in R:\xi_{ik} \neq 0 \text{ for some }i\in Q_{C}\left(R\right)\right\},
\end{align*}
then the stationary distribution of the subnetwork
can be used to compute the propensities of the reactions in $R_{S}\left(R\right)$.
This is called the quasi-stationary assumption \cite{Cao2005}.

We now specify a formal criteria to check if the
subnetwork dynamics is sufficiently fast in comparison
to its surrounding network.
\trred{
For this we define a timescale for each reaction $k$ as
\[
\zeta\left(k\right) =
	\frac{\log \kappa'_{k}}{\log N_{0}}
	+ \sum_{i = 1}^{n_{S}}
		\nu_{ik}
		\cdot \frac{\log x_{i} }{\log N_{0}}
		\cdot \mathds{1}_{ \left\{  x_{i} > 0 \right\} }
\\
\]
and the observation timescale as
$\zeta\left(0\right) = \frac{\log\left(1 / t_f\right)}{\log{N_{0}}}$.
The form of $\zeta\left(k\right)$ is chosen in such a way that $N_{0}^{\zeta\left(k\right)}$ captures the timescale of the $k$-th reaction channel.
However, when reaction channel $k$ consumes a species $i$ (i.e.\ $\xi_{ik} < 0$) and the copy-number $x_{i} = 0$, then the exponent to $N_{0}$ would have to be
$- \infty$ to exactly capture the timescale of the $k$-th reaction channel. A single reaction producing species $i$ would then change the timescale to a finite value
and immediately render the averaging decision based on the timescale as invalid.
To prevent this we overestimate the timescale of a reaction channel by ignoring the $-\infty$ terms when computing $\zeta\left(k\right)$.
This is achieved by adding the indicator functions $\mathds{1}_{ \left\{  x_{i} > 0 \right\} }$ in $\zeta\left(k\right)$.
}


The timescale separation of the fast subnetwork is then given by
\begin{align}
\label{eq:averaging-timescale-separation}
\Delta\zeta\left(R\right) = \min_{k\in R}\left(\zeta\left(k\right)\right) - \max_{k\in R_{S}\left(R\right) \cup \left\{0\right\}}\left(\zeta\left(k\right)\right) \quad .
\end{align}
In words, $\Delta\zeta\left(R\right)$ is the difference between the slowest timescale of all reactions within the subnetwork
and the fastest timescale of all reactions connecting the subnetwork to the surrounding network.
We fix a positive parameter $\Theta$ (with $\Theta = 0.5$ as the default value).
If $\Delta \zeta \left( R \right) \geq \Theta$
then we call the subnetwork defined by $R$ and $Q\left(R\right)$ to be \textit{fast}.

Given that a subnetwork is fast, we can apply the quasi-stationary assumption if the dynamics of this subnetwork
converges to a stationary distribution.
Below we describe three different strategies to check if this is indeed the case.

\label{subsec:averaging_strategies}
\subsection*{Finite Markov Chains}

Suppose that the species in $Q\left(R\right)$ satisfy a conservation relation of the form
$\sum_{i \in Q\left(R\right)} b_{i} x_{i} = C$,
where the $b_{i}$-s and $C$ are positive integers
and assume that the copy-numbers of species outside $Q\left(R\right)$ are fixed.
Due to the reactions in $R$ the copy-number vector of all the species is constrained
to reside in a finite set, which can be enumerated as
$\left\{e_1,\ldots,e_L\right\}$.
\trred{
Define a $L\times L$ matrix $A$ with entries given by
\[
A_{ij} = \begin{cases}
\sum_{R_{ij}} \lambda'_{k}\left( e_{i} \right) & \text{ if } i \neq j \\
-\sum_{k \in R} \lambda'_{k}\left( e_{i} \right) & \text{ if }i=j \\
\end{cases}
\]
where $R_{ij} = \left\{k \in R: e_{i} + \xi_{k} = e_{j}\right\}$.
}
If there exists a unique positive vector $\pi = \left(\pi_1,\ldots,\pi_L\right)$ such that
$\sum_{i=1}^{L} \pi_i = 1$ and $\pi A = 0$, then the Markov process corresponding to the
subnetwork dynamics is ergodic with the unique stationary distribution given by the vector $\pi$.
One can check the existence and uniqueness of $\pi$ by verifying that the dimension
of the left null-space of $A$ is one \cite{seneta2006}.

\subsection*{Pseudo-Linear subnetworks}

Consider the subnetwork defined by $R$ and $Q\left(R\right)$ as before.
Assume that the propensities of the reactions in $R$ are affine functions of the
copy-numbers of the species in $Q\left(R\right)$,
that is $\sum_{i\in Q\left(R\right)} \nu_{ik} \leq 1$ for all $k \in R$.
Such a subnetwork is called \textit{pseudo-linear} because it does not involve any bimolecular reactions
between the fast species in the subnetwork.
For linear networks the moment equations are closed. Hence we 
can easily find the first two stationary moments by
finding the fixed points of the moment equations
\begin{align*}
\frac{d\mathbb{E}[X_{i}]}{dt} = & \sum_{k\in R} \xi_{ik} \mathbb{E}[\lambda'_{k} \left(X\right)] \\
\frac{d\mathbb{E}[X_{i} X_{j}]}{dt} = & \sum_{k\in R} \xi_{ik} \mathbb{E}[\lambda'_{k} \left(X\right) X_{j}] + \sum_{k\in R} \mathbb{E}[\lambda'_{k} \left(X\right) X_{i}] \xi_{jk} 
	+ \sum_{k\in R} \xi_{ik} \mathbb{E}[\lambda'_{k} \left(X\right)] \xi_{jk},
\end{align*}
and checking if these fixed points are stable. As we only consider at most bimolecular reactions we only need the first two stationary moments to apply the quasi-stationary assumption.

\subsection*{Zero-Deficiency subnetworks}

In some situations, a fast subnetwork may not be pseudo-linear
and its state space is either infinite or too large to be able to compute the exact
stationary distribution for the corresponding Markov Chain (for example see the \textit{Fast Dimerization}
network in Section \ref{sec:NUMERICAL-EXAMPLES}).
In such cases it is sometimes possible to use the recent
results from Anderson et al.\cite{anderson_kurtz_2010}\ to compute
the stationary distribution of the subnetwork dynamics.
In particular Anderson et al.\cite{anderson_kurtz_2010}\ show
that \textit{weakly reversible} \footnote{
A SRN is called \textit{weakly reversible} if for every reaction $k$
there is a sequence of reactions $\nu_{k_{1}}\rightarrow\nu'_{k_{1}}\rightarrow\cdots\rightarrow\nu_{k_{m}}\rightarrow\nu'_{k_{m}}$
such that $\nu'_{k}=\nu_{k_{1}}$ and $\nu_{k}=\nu'_{k_{m}}$.
}
networks with zero \textit{deficiency} \footnote{
The \textit{deficiency} of a SRN is defined as $\delta=|C|-l-s$ where
$|C|$ is the number of reaction complexes ($C=\{\nu_{k}\}\cap\{\nu'_{k}\}$),
$l$ is the number of linkage classes (a linkage class is a connected
component of the reaction complex graph corresponding to the SRN)
and $s$ is the dimension of the stoichiometric subspace $S=span_{k\in\left\{ 1,\ldots,n_{R}\right\} }\left\{ \xi_{k} \right\} $.
We refer the readers to \cite{anderson_kurtz_2010} for more details.
}
admit a product-form stationary distribution given by
\trred{\begin{align}
\label{eq:product_form_stationary_distribution}
\pi\left(x\right) =
	M\prod_{i\in Q\left(R\right)}\frac{c_{i}^{x_{i}}}{x_{i}!}, \quad & x \in \Gamma
\end{align}
and $\pi\left(x\right) = 0$ otherwise,
where $\Gamma$ is an irreducible state space containing the initial state,
$c\in\mathbb{R}_{\geq0}^{n_{S}}$ is the equilibrium point of
the corresponding deterministic system and $M$ is a positive normalizing
constant.
We can find the appropriate set $\Gamma$ using the results in \cite{gupta2013determining},
and this would imply that the dynamics of the distribution of the SRN will converge to the stationary distribution $\pi$.
}

\subsection*{Implementation details}

During the simulation of the sample path we need to keep track of
changes to the copy-numbers of species in $Q\left(R\right)$,
so that we can recompute the moments of the stationary distribution of the fast subnetwork.
We do this by modifying the partitioning of the reactions, so that the reactions
in $R_{S}\left(R\right)$ are discrete.
Every time a reaction in $k\in R_{S}\left(R\right)$ occurs, the
first and second moments of the stationary distribution are recomputed
and the appropriate rate constants are updated.
For this, Algorithm 2 has to be modified in a straightforward manner.
Below we describe the averaging procedure.

\medskip

\textbf{Algorithm 4: Averaging}

\begin{my_algorithm}
	\item \textbf{Once:} Precompute subnetworks $L_A$ suitable for averaging \\(e.g. weakly reversible, zero-deficiency subnetworks)
	\item Let $L_P$ be the set of previously averaged subnetworks
	\item Set $L_C = \emptyset$
	\item \textbf{for each} suitable subnetwork $R \in L_A$ \textbf{do}
		\item \ \textbf{if} $\Delta\zeta\left(R\right) \geq \Theta$ \textbf{then}
			\item \ \ Set $L_C = L_C \cup \left\{R\right\}$
		\item \ \textbf{end if}
	\item \textbf{end for}
	\item $L_F = \emptyset$
	\item \textbf{while} $L_C \neq \emptyset$ \textbf{do}
		\item \ Set $R = \argmax_{W\in L_C} \left(|W|\right)$
		\item \ Set $L_C = L_C \setminus {R}$
		\item \ \textbf{if} $\underset{W\in L_{F}}{\cup} W \cap R = \emptyset$ and \\
			$\underset{W\in L_{F}}{\cup} Q\left(W\right) \cap Q\left(R\right) = \emptyset$ \textbf{then}
			\item \ \ Set $L_F = L_F \cup \left\{R\right\}$
			\item \ \ Modify $R_{D}, R_{C}$, so that $R_{S}\left(R\right) \subseteq R_{D}$
		\item \ \textbf{end if}
	\item \textbf{end while}
	\item \textbf{for each} subnetwork $R \in L_P \setminus L_F$ \textbf{do}
		\item \ Sample state of the species in $Q\left(R\right)$ from the \\stationary distribution
	\item \textbf{end for}
	\item \textbf{for each} subnetwork $R\in L_F$ \textbf{do}
		\item \ Compute stationary first and second moment of the \\species in $Q\left(R\right)$
		\item \ Update the stoichiometries and propensities of \\reactions in $R_{S}\left(R\right)$
	\item \textbf{end for}
\end{my_algorithm}

\section{Mathematical Justification\label{appendix:MATHEMATICAL_JUSTIFICATION}}

In this section we provide a mathematical justification of
why our adaptive scheme produces a close approximation to the solution
of a CME.
We start with a simple lemma that shows convergence of a sequence of
processes, in which each process is formed by stitching together
two Markov processes at a random stopping time determined by the first process.
In this section let $n = n_{S}$, where $n_{S}$ is the number of species,
and let ``$\Rightarrow$'' and ``$\approx_{d}$'' denote convergence in distribution
and similarity in distribution respectively.
Let $\mathbb{D}_{\mathbb{R}^{n}} [0,\infty)$ denote
the space of c\`adl\`ag functions, from $[0,\infty)$ to $\mathbb{R}^{n}$,
that is, right continuous functions with
left limits \cite{ethier2009markov}.
Let $A$ be a measurable subset of $\mathbb{R}^{n}$ and define
$\mathcal{G}_{A}: \mathbb{D}_{\mathbb{R}^{n}} \left[0,\infty\right) \rightarrow [0,\infty]$ by
\begin{align}
\label{eq:stopping_time}
\mathcal{G}_{A}\left(\zeta\right) = \inf\left\{t \geq 0: \zeta\left(t\right) \not\in A\right\}.
\end{align}
For each $N$, let
$W_{1}^{N}, W_{2}^{N}$ be two Markov processes with sample paths in $\mathbb{D}_{\mathbb{R}^{n}} \left[0,\infty\right)$.
Assume that for any sequence of initial conditions satisfying $W_{i}^{N}\left(0\right) \Rightarrow W_{i}\left(0\right)$ we have
$W_{i}^{N} \Rightarrow W_{i}$ as $N \rightarrow \infty$, where $W_{i}$ is a PDMP as in \eqref{eq:random_time_PDMP}.
Define a stopping time $\tau^{N} = \mathcal{G}_{A} \left(W_{1}^{N}\right)$.
Suppose $W_{1}^{N}\left(0\right) = w_{0}$ and
$W_{2}^{N}\left(0\right) = f\left( W_{1}^{N}\left(\tau^{N}\right) \right)$,
where $f: \mathbb{R}^{n} \rightarrow \mathbb{R}^{n}$ is a measurable function.

Define another process $W^{N}$ by
\[
W^{N}\left(t\right) =
W_{1}^{N}\left(t\right) \mathds{1}_{\left\{t < \tau^{N}\right\}}
+ W_{2}^{N}\left(t - \tau^{N}\right) \mathds{1}_{\left\{t \geq \tau^{N}\right\}}
\]

\newtheorem*{lemma}{Lemma}
\begin{lemma}
For any $t \geq 0$ we have
\[
W^{N}\left(t\right) \Rightarrow W\left(t\right)
\]
where
\[
W\left(t\right) =
W_{1}\left(t\right) \mathds{1}_{\left\{t < \tau\right\}}
+ W_{2}\left(t - \tau\right) \mathds{1}_{\left\{t \geq \tau\right\}}
\]
with $\tau = \mathcal{G}_{A}\left(W_{1}\right)$, $W_{1}\left(0\right) = w_{0}$ and $W_{2}\left(0\right) = f\left(W_{1}\left(\tau\right)\right)$.
\end{lemma}

\begin{proof}
Let $g : \R^n \to \R^n$ be any bounded continuous function. To prove the
lemma it suffices to show that for any $t \geq 0$
\begin{align}
\label{lem:condsuff}
\lim_{N \to \infty} \E ( g( W^N(t) ) ) = \E( g( W(t) )),
\end{align}
where $\E$ denotes the expectation operator. Fix a $t \geq 0$. We can write
\begin{align}
\label{exp_split}
\E \left( g( W^N(t) ) \right)
& = \E  \left( g( W^N(t) )  \mathds{1}_{ \{ t <
\tau^N\} } + g( W^N(t) )  \mathds{1}_{ \{ t \geq \tau^N\} }
\right) \\
& =   \E  \left( g( W^N_1(t) )  \mathds{1}_{ \{ t < \tau^N\} } + g( W^N_2(t - \tau^N) )  \mathds{1}_{ \{ t \geq \tau^N\} } \right). \notag
\end{align}
Note that $\tau^N = \mathcal{G}_A(W^N_1)$ and $W^N_1 \Rightarrow W_1$,
where $W_1$ is a PDMP. Hence, we can conclude that $(W^N_1,\tau^N)
\Rightarrow (W_1,\tau)$, where $\tau = \mathcal{G}_A (W_1)$. This shows
that
\begin{align}
\label{exp_split1}
\lim_{N \to \infty} \E  \left( g( W^N_1(t) )  \mathds{1}_{ \{ t < \tau^N\} }
\right) = \E  \left( g( W_1(t) )  \mathds{1}_{ \{ t < \tau\} } \right).
\end{align}
The second term in \eqref{exp_split} can be expressed as
\begin{align*}
\E  \left( g( W^N_2(t - \tau^N) )  \mathds{1}_{ \{ t \geq \tau^N\} } \right) =
\E \left( \E\left( g( W^N_2(t - \tau^N) )  \mathds{1}_{ \{ t \geq \tau^N\} }
\vert W^N_1(\tau^N) \right) \right)
\end{align*}
Given $W^N_1(\tau^N) =w$, the process $ W^N_2$ converges in distribution
to the PDMP $W_2$ with initial state $f(w)$. Moreover, the results in
\cite{baxter1977compactness} show that $W^N_1(\tau^N) \Rightarrow
W_1(\tau)$. Therefore
\begin{align*}
\lim_{N  \to \infty} \E  \left( g( W^N_2(t - \tau^N) )  \mathds{1}_{ \{ t \geq
\tau^N\} } \right) & =
\E \left( \E\left( g( W_2(t - \tau) )  \mathds{1}_{ \{ t
\geq \tau\} }  \vert W_1(\tau) \right) \right) \\ & =
\E\left( g( W_2(t - \tau)
)  \mathds{1}_{ \{ t \geq \tau\} } \right).
\end{align*}
Using this relation along with \eqref{exp_split1} and \eqref{exp_split} we
obtain
\begin{align*}
\lim_{N \to \infty} \E \left( g( W^N(t) ) \right) 
= \E  \left( g( W_1(t)
)  \mathds{1}_{ \{ t < \tau\} } \right) + \E\left( g( W_2(t - \tau) )  \mathds{1}_{ \{
t \geq \tau\} } \right) = \E( g( W(t) )).
\end{align*}
This shows \eqref{lem:condsuff} and finishes the proof of this lemma.
\end{proof}

The above lemma can be easily generalized to the more general case, where $W^{N}$ is defined
by stitching together $m$ Markov processes
$W_{1}^{N},\ldots,W_{m}^{N}$, where each $W_{i}^{N}$ converges to a PDMP $W_{i}$.
We define $\tau_{0}^{N} = 0$ and $\tau_{i}^{N} = \tau_{i-1}^{N} + \mathcal{G}_{A} \left( W_{i}^{N} \right)$,
and similarly $\tau_{0} = 0$ and $\tau_{i} = \tau_{i-1} + \mathcal{G}_{A} \left( W_{i} \right)$ for $i=1,\ldots,(m - 1)$.
The initial states of the processes $W_{i}^{N}$ and $W_{i}$ are given by
$W_{i}^{N}\left(0\right) = f_{i} \left( W_{i-1}^{N}\left(\tau_{i-1}^{N}\right)\right)$
and $W_{i}\left(0\right) = f_{i} \left( W_{i-1}\left(\tau_{i-1}\right)\right)$ respectively, for $i=2,\ldots,m$,
where $f_{i}: \mathbb{R}^{n} \rightarrow \mathbb{R}^{n}$ are measurable functions.
Assuming $W^{N}_{1}\left(0\right) \Rightarrow W_{1}\left(0\right)$, the process $W^{N}$ given by
\begin{align}
\label{eq:stitched_markov_processes}
W^{N}\left(t\right) =
\sum_{i=1}^{m} W_{i}^{N}\left(t - \tau_{i-1}^{N} \right) \mathds{1}_{\left\{\tau_{i-1}^{N} \leq t < \tau_{i}^{N}\right\}}
\end{align}
converges in distribution to the process $W$ defined by
\begin{align}
\label{eq:stitched_PDMP_processes}
W\left(t\right) =
\sum_{i=1}^{m} W_{i}\left(t - \tau_{i-1}\right) \mathds{1}_{\left\{\tau_{i-1} \leq t < \tau_{i}\right\}}.
\end{align}

In our adaptive scheme the times of adaptation are defined analogously to \eqref{eq:stopping_time} where the set
$A$ is given by \eqref{eq:copy_number_bounds}.
We denote the times of adaptation by
$\tau_{i}^{N}$ for $i=1,\ldots,(m-1)$, where $(m-1)$ is the number of adaptations.
We define $\tau_{0}^{N} = 0$ and $\tau_{m}^{N} = \infty$ for convenience.
The computation of the scaling parameters $\alpha = \left(\alpha_{1}, \ldots, \alpha_{n}\right)$ and $\beta = \left(\beta_{1}, \ldots, \beta_{n_{R}}\right)$
only depends on the state of the simulated process at the time of adaptation.
Suppose that this computation is given by functions
$L_{\alpha}, L_{\beta}: \mathbb{R}^{n} \rightarrow \mathbb{R}^{n}$, such that
for $\tau_{i}^{N} \leq t < \tau_{i+1}$ we have
$\alpha\left(t\right) = L_{\alpha} \left(W^{N}\left(\tau_{i}^{N}\right)\right)$ and $\beta\left(t\right) = L_{\beta} \left(W_{i}^{N}\left(\tau_{i}^{N}\right)\right)$.
Here $W^{N}$ denotes the adaptive version of the scaled family of processes in \eqref{eq:multiscale_random_time_change}.
We can write $W^{N}$ in the form \eqref{eq:stitched_markov_processes} with
\begin{align*}
W_{i}^{N}\left(t\right) & = W^{N}_{i}\left(0\right) + N^{-\alpha\left(\tau^{N}_{i-1}\right)}
\sum_{k=1}^{n_R}Y_{k}\left( \int_{0}^{t}N^{\beta\left(\tau^{N}_{i-1}\right)+\alpha\left(\tau^{N}_{i-1}\right)\cdot\nu_{k}}\lambda_{k}^{N}\left(W_{i}^{N}(s\right)ds\right)\xi_{k}
\end{align*}
where $N^{x} = \diag\left(N^{x_1},\ldots,N^{x_k}\right)$ for a $k$-dimensional vector $x = \left(x_1,\ldots,x_k\right)$.
The initial conditions are given by
$W^{N}_{1}\left(0\right) = N_{0}^{-\alpha\left(0\right)} X\left(0\right)$ and
$W^{N}_{i}\left(0\right) = N_{0}^{\alpha\left(\tau^{N}_{i-1}\right)-\alpha\left(\tau^{N}_{i}\right)} W_{i}^{N}\left(\tau^{N}_i\right)$ for $i=2,\ldots,m$.
From Kang et al.\cite{kang2013separation}\ we know that each $W^{N}_{i}$ converges in distribution
to a PDMP $W_{i}$, as described in Section \ref{sec:multiscale_framework}.
From the previous lemma it follows that $W^{N}\left(t\right) \Rightarrow W\left(t\right)$ for any $t \geq 0$,
where $W$ has the form \eqref{eq:stitched_PDMP_processes}
with the $W_{i}$-s given by the convergence result of Kang et al.\cite{kang2013separation}, given in \ref{sec:multiscale_framework}.

Now we return to our original process $X$ from Section \ref{sec:multiscale_framework} and recall that for large $N_{0}$
we have $X\left(t\right) = N_{0}^{\alpha} Z^{N_{0}}\left(t\right) \approx_{d} N_{0}^{\alpha} Z\left(t\right)$.
With adaptation we also have $X\left(t\right) = N_{0}^{\alpha} W^{N_{0}}\left(t\right)$ and as
$W^{N}\left(t\right) \Rightarrow W\left(t\right)$ for any $t \geq 0$, it follows that
$X\left(t\right) = N_{0}^{\alpha} W^{N_{0}}\left(t\right) \approx_{d} N_{0}^{\alpha} W\left(t\right)$ and this justifies
the use of our adaptive scheme to approximate the distribution of the original process $X$ at any fixed time $t \geq 0$.



\bibliographystyle{ieeetr}
\bibliography{arxiv_paper}


\end{document}